\documentclass[journal, letterpaper, 10pt, romanappendices]{IEEEtran}
\usepackage{slashbox}
\usepackage{graphics}
\usepackage{cite}
\usepackage[pdftex]{graphicx}
\usepackage{epstopdf}
\usepackage{epsfig}
\usepackage{latexsym}
\usepackage{amsfonts}
\usepackage{calc}
\usepackage{url}
\usepackage{enumerate}
\usepackage{color}
\usepackage[tbtags]{amsmath}
\usepackage{amssymb}
\usepackage{upref}
\usepackage{dsfont}
\usepackage{multirow}
\usepackage{booktabs}
\usepackage{bigstrut}
\usepackage{rotating}
\usepackage{nth}
\usepackage{breqn}

\newtheorem{theorem}{Theorem}

\newtheorem{lemma}{Lemma}

\newtheorem{proof}[theorem]{Proof}

\begin{document}
\title{Efficient Feedback-Based Scheduling Policies for Chunked Network Codes over Networks with Loss and Delay}

\author{\IEEEauthorblockN{Anoosheh~Heidarzadeh and Amir H. Banihashemi}\\
\IEEEauthorblockA{\small{Department of Systems and Computer Engineering, Carleton University, Ottawa, ON, Canada}
Email: \texttt{\{anoosheh,ahashemi\}@sce.carleton.ca}}
}

\maketitle
\thispagestyle{empty}
\begin{abstract} The problem of designing efficient feedback-based scheduling policies for chunked codes (CC) over packet networks with delay and loss is considered. For networks with feedback, two scheduling policies, referred to as \emph{random push} (RP) and \emph{local-rarest-first} (LRF), already exist. We propose a new scheduling policy, referred to as \emph{minimum-distance-first} (MDF), based on the expected number of innovative successful packet transmissions at each node of the network prior to the ``next'' transmission time, given the feedback information from the downstream node(s) about the received packets. Unlike the existing policies, the MDF policy incorporates loss and delay models of the link in the selection process of the chunk to be transmitted. Our simulations show that MDF significantly reduces the expected time required for all the chunks (or equivalently, all the message packets) to be decodable compared to the existing scheduling policies for line networks with feedback. The improvements are particularly profound (up to about $46\%$ for the tested cases) for smaller chunks and larger networks which are of more practical interest. The improvement in the performance of the proposed scheduling policy comes at the cost of more computations, and a slight increase in the amount of feedback. We also propose a low-complexity version of MDF with a rather small loss in the performance, referred to as \emph{minimum-current-metric-first} (MCMF). The MCMF policy is based on the expected number of innovative packet transmissions prior to the ``current'' transmission time, as opposed to the next transmission time, used in MDF. Our simulations (over line networks) demonstrate that MCMF is always superior to RP and LRF policies, and the superiority becomes more pronounced for smaller chunks and larger networks. We also compare the performances of the existing RP and LRF policies, and show that their relative performance (including which one performs better) depends on delay and loss models, the network length and the chunk size.\end{abstract}



\section{Introduction}
There has recently been a surge of interest in the application of coding schemes over packet networks, e.g., for large-scale file sharing~\cite{GR:2005,GMR:2006,GMR2:2006,NL:2007}. In particular, random linear network codes (dense codes) are known to reduce the expected \emph{delivery time}\footnote{For a given code over a given network, ``delivery time'' is defined as the minimum time required for communicating the message(s) of the source node(s) to the sink node(s) throughout the network.} in comparison to routing protocols over networks with arbitrary link delays and erasures \cite{MHL:2006}. This, however, comes at the cost of large computational complexity of the coding algorithms. To reduce the coding cost of dense codes, \emph{chunked codes} (CC) and \emph{overlapped chunked codes} (OCC) were proposed in \cite{MHL:2006,SZK:2009,HB:2010,HB:2011}. These codes operate by dividing the original message at the source node into non-overlapping or overlapping chunks, respectively, and each non-sink network node schedules the transmission of the chunks at random by using a dense code. The coding cost of these codes are linear in the size of the chunks, smaller than that of dense codes in general. This however comes at the expense of larger expected delivery time.

Originally, CC and OCC were designed for and analyzed over arbitrary network realizations\footnote{Here, we use the term ``network realization'' to refer to a member of the ensemble of networks with random link erasures and random link delays.} (worst-case analysis) in the absence of feedback \cite{MHL:2006,SZK:2009,HB:2010,HB:2011}. In real-world scenarios, however, feedback is often available. One thus expects to reduce the expected delivery time when the feedback is properly used. In other words, the scheduling of chunks uniformly at random, referred to as the \emph{random scheduling policy}, might result in wasting a large number of transmission opportunities. The reason is that, such a scheme treats those chunks which are already decodable or are short of only a few more packets to be decodable, similar to those chunks which need a much larger number of packets to be decodable. The problem is therefore how to use feedback and devise a scheduling policy (for CC)\footnote{CC are the focus of this paper, and in the case of OCC, the generalization of the proposed scheduling policies is not trivial, and is beyond the scope of this paper.} which outperforms the random scheduling policy.\footnote{It should be noted that routing itself is a special case of chunked coding with the number of chunks equal to the number of message packets at the source node. On the other hand, the design of efficient feedback-based scheduling policies for routing over networks with delay and loss is still an open problem. Thus, the scheduling policies proposed in this paper can also be used for distributed routing over any network topology.}

In earlier related works \cite{WL:2007,XZWX:2008}, two general policies, which utilize the feedback information to schedule the chunks, were proposed. These scheduling policies were referred to as \emph{random push} (RP) and \emph{local-rarest-first} (LRF), respectively. Both RP and LRF scheduling policies, employed by the transmitting node over a link, use the number of \emph{innovative packets}\footnote{A packet is said to be ``innovative'' at a node if its global encoding vector (i.e., the vector of the coefficients which represent the mapping between the packet and the message packets at the source node) is linearly independent of the global encoding vectors of the packets previously received by the node.} which have been received by the receiving node of the link till the current transmission time. In RP \cite{WL:2007}, the node transmitting over a link chooses a chunk uniformly at random from the set of chunks that still need more innovative packets to be decodable at the receiving node of the link. In LRF \cite{XZWX:2008}, however, the transmitting node chooses a chunk which needs the largest number of innovative packets at the receiving node.

In both RP and LRF policies, at each time instant, a transmitting node makes a decision based on the set of received packets at the receiving node up to that point in time. In the presence of delay, however, such a decision fails to take into account the contribution of the (successful) packets that were transmitted earlier to the receiving node (over the same link or the other links with different transmitting nodes but with the same receiving node as the underlying one) but have not still been received due to the delay. One thus expects to be able to improve these scheduling polices over the networks with delay. Related to this, one should note that both RP and LRF policies utilize the feedback information in order to count the number of innovative packets delivered to the receiving node. This, however, disregards the packets which have been (successfully) transmitted but still have not been received. Nevertheless, the more are such transmissions corresponding to a chunk, the larger is the probability of delivering more useful information about the underlying chunk. In addition, thanks to the literature on modeling the packet loss and the packet delay over networks with feedback (e.g., see \cite{Paxson:1997,Moon:2000} and references therein), such probabilities can be computed with a reasonably high accuracy. This however comes at the cost of more computation at the network nodes. In this paper, we do not focus on the problem of modeling the loss and the delay of the network links, the estimation of the model parameters, and the tradeoff between the accuracy and the computational complexity. Throughout this paper, we assume that the models of the packet loss and the packet delay of each link are known at the transmitting/receiving nodes of the link. The question then is how to properly use (i) the knowledge about the sets of transmitted and received packets over a link, (ii) the knowledge about the sets of received packets over the rest of the links with the same receiving node (as that of the underlying link), and (iii) the knowledge about the link model parameters, in order to decrease the expected delivery time. In an attempt to answer this question, the main contributions of this work are as follows:

\begin{itemize}
\item{We propose a new scheduling policy for chunked codes, referred to as \emph{minimum-distance-first} (MDF), devised based on a new metric, i.e., the expected number of innovative packets transmitted prior to the \emph{next} transmission time.}
\item{Aiming at the design of a low-complexity version of MDF, we also propose another scheduling policy for chunked codes, referred to as \emph{minimum-current-metric-first} (MCMF), which works based on the expected number of innovative packets transmitted prior to the \emph{current} transmission time.}
\item{We show through extensive simulations over line networks (as the simplest non-trivial network topology for the unicast problem\footnote{In a practical scenario, the line network topology would be the right model for an overlay network where the sequence of nodes are determined by an underlying routing protocol.}) that (i) the MDF scheduling policy performs (near) optimal in the sense of minimizing the expected delivery time; (ii) both MDF and MCMF are always superior to LRF and RP with respect to the expected delivery time, and that the improvements are particularly large for smaller chunks and larger networks as well as delays with smaller mean and variance; (iii) MCMF is always inferior to MDF, but the performance loss becomes smaller for larger chunks, smaller networks and delays with smaller mean and variance; (iv) the relative performance of MDF or MCMF compared to the random scheduling policy depends on the delay distribution. In particular, the advantage of the proposed scheduling policies becomes more profound for smaller chunks, larger networks and delays with smaller mean and variance; (v) for sufficiently small chunks and sufficiently large networks, RP is superior to LRF, and the advantage is more evident for smaller chunks, larger networks, and for delays with larger mean and variance.}
\end{itemize}

\vspace{-.15 cm}
\section{Model and Assumptions}\label{sec:MA}
\subsection{Network Topology}\label{subsec:NetworkTopology}
We consider a unicast problem over a network with $L$ (directed) links with any arbitrary topology, where one \emph{source} node (which is not the receiving node of any link) which possesses $k$ message packets, each a string of bits, and one \emph{sink} node (which is not the transmitting node of any link) which demands the message packets, are connected through the rest of the network nodes, called \emph{internal} nodes.

We also consider an arbitrary ordering of the $L$ links in the network, and associate a label (i.e., a unique integer in $\{1,\ldots,L\}$) to each link. For every $1\leq i\leq L$, let $\mathcal{I}^{(i)}_{R}$ (or $\mathcal{I}^{(i)}_{T}$) be the set of labels of the links whose receiving nodes (or transmitting nodes) are the same as the receiving node (or the transmitting node) of the $i\textsuperscript{th}$ link.

\vspace{-.25 cm}
\subsection{Loss and Delay Models}\label{subsec:ScheduleModel}
In the following, we describe the loss and the delay models used in this work. One, however, should note that the application of the scheduling policies discussed in this paper is not restricted to a specific model of delay or loss.

Each link is modeled by a memoryless erasure channel with a constant probability of erasure, i.e., for every $1\leq i\leq L$, the $i\textsuperscript{th}$ link has a probability of erasure $p^{(i)}_e$, for some $0\leq p^{(i)}_e\leq 1$ (each packet transmitted over the $i\textsuperscript{th}$ link is either erased with probability $p^{(i)}_e$, or is successfully received with probability $1-p^{(i)}_e$). We also assume that the links are affected by erasures independently. The special case with no erasure (i.e., $p^{(i)}_e=0$, for all $i$) is referred to as the \emph{lossless} case.

Each successful (not erased) packet transmitted at time $n$ over the $i\textsuperscript{th}$ link is assumed to experience a delay $Z^{(i)}_n\in \mathds{Z}^{+}\setminus\{0\}$, i.e., the packet arrives at time $n+Z^{(i)}_n$, where $Z^{(i)}_n$ is a random variable with the probability mass function \begin{equation}\label{EQ2}P_{Z^{(i)}_n}[z]=\int_{z-1}^{z}{f_{R^{(i)}_n}(r)dr},\end{equation} for every $z\in\mathds{Z}^{+}\setminus\{0\}$, where $f_{R^{(i)}_n}(r)$, $r\in\mathds{R}^{+}$, is a probability density function. Note that $Z^{(i)}_n$ is a discrete version of the continuous random variable $R^{(i)}_n$. For all $n$, $Z^{(i)}_n$'s are assumed to be independent and identically distributed.\footnote{Here, without loss of generality, we have assumed that the time unit is equal to the inverse of the packet transmission rate at each network node.} The special case with all the delays equal to $1$ (i.e., $Z^{(i)}_n=1$, for all $i$ and $n$) is also referred to as the \emph{unit-delay} model.


\subsection{Information Available at Network Nodes}\label{subsec:Assumptions}
Each node is assumed to: (i) know the loss and delay models (called the \emph{link model}) of each link over which it transmits a packet, and (ii) store all the packets it transmits/receives along with their departure/arrival times. In particular, each node keeps the record of all the packets it transmits. Moreover, right after the reception of a new packet, each node stores the packet if the packet is innovative to the set of all its previously received innovative packets, or discards the packet, otherwise. Note that, in the case of transmitted packets, it suffices that each node stores the global encoding vector of each packet included in the packet header (which is often much smaller than the packet payload), instead of storing both the packet header and the packet payload. In the case of the received packets, both the packet header and the packet payload need to be stored. This is not however a burden when the internal nodes also demand all the message packets (e.g., in the application of peer-to-peer file sharing).


\section{Problem Statement}\label{sec:ProblemStatement}
In CC, the $k$ message packets, at the source node, are divided into $q$ disjoint subsets, called \emph{chunks}, each of size $k/q$. Each non-sink node, at each time $n$, chooses a chunk, say $\omega\in [q]:=\{1,\dots,q\}$, based on a scheduling policy, and by applying a specific coding algorithm to its previously received packets pertaining to chunk $\omega$ (\emph{$\omega$-packets}\footnote{For every $\omega\in[q]$, a packet is called an ``$\omega$-packet'' if it can be written as a linear combination of the message packets belonging to the chunk $\omega$.}) generates/transmits a new $\omega$-packet. The sink node is able to decode the chunk $\omega$ so long as it receives $k/q$ innovative $\omega$-packets.

Let $\mathcal{T}^{(i)}_n$ and $\mathcal{R}^{(i)}_n$ be the set of packets transmitted and received over the $i\textsuperscript{th}$ link till time $n$, respectively, and $\mathcal{T}^{(i)}_n (\omega)$ and $\mathcal{R}^{(i)}_n (\omega)$ be the set of the $\omega$-packets in $\mathcal{T}^{(i)}_n$ and $\mathcal{R}^{(i)}_n$, respectively. Note that, by the assumption made in Section~\ref{subsec:Assumptions}, all the $\omega$-packets in $\mathcal{R}^{(i)}_n (\omega)$ are innovative, and none of the $\omega$-packets in $\mathcal{T}^{(i)}_n(\omega)\setminus\mathcal{R}^{(i)}_n (\omega)$ is received yet.

Based on the presence or absence of feedback in the network, one can devise different scheduling policies. If no feedback is available, for all $1\leq i\leq L$, at time $n$, the transmitting node of the $i\textsuperscript{th}$ link knows $\bigcup_{j\in \mathcal{I}^{(i)}_{T}} \mathcal{T}^{(j)}_n$, but has no information about $\mathcal{R}^{(j)}_n$, for any $j\in \mathcal{I}^{(i)}_{R}$. However, in the presence of feedback, whenever a packet arrives, the receiving node sends a delay-free and error/erasure-free acknowledgment to the transmitting node, along with a message containing information about the departure time of the arrived packet. The receiving node will also send messages to all the transmitters of the links in $\mathcal{I}^{(i)}_{R}$ to convey information about the received packet. In addition to being delay-free, the feedback channels are assumed to have no error/erasure.\footnote{It should be noted that the assumptions of delay-free and error/erasure-free feedback are reasonable because the data rate over the channel used for feedback is often very low compared to that of the channels used for forward packet transmission.} Thus, in the presence of feedback, the transmitting node has full knowledge of $\bigcup_{j\in \mathcal{I}^{(i)}_{T}} \mathcal{T}^{(j)}_n$ and $\bigcup_{j\in \mathcal{I}^{(i)}_{R}} \mathcal{R}^{(j)}_n$.

The problem, at the transmitting node of the $i\textsuperscript{th}$ link, for every $i$, at every time $n$, is how to select a chunk and to code it over the $i\textsuperscript{th}$ link, given the link model, in order to minimize the expected delivery time (where the expectation is taken over all the realizations of the code and the network), i.e., the expected time required for all the chunks to be decodable, when only $\bigcup_{j\in \mathcal{I}^{(i)}_{T}} \mathcal{T}^{(j)}_n$, or both $\bigcup_{j\in \mathcal{I}^{(i)}_{T}} \mathcal{T}^{(j)}_n$ and $\bigcup_{j\in \mathcal{I}^{(i)}_{R}} \mathcal{R}^{(j)}_n$ are known.

\section{Existing Solutions}\label{sec:ES}
\subsection{Random Scheduling Policy}
Originally, CC were designed for networks with no feedback~\cite{MHL:2006}. In this scenario, one possible strategy for a transmitting node is to use a \emph{fully random scheduling policy}, specified as follows: The node chooses a chunk, say $\omega$, uniformly at random; if the node is source, it generates/transmits a random linear combination of all the packets belonging to the chunk $\omega$, and if it is internal, it generates/transmits a random linear combination of all its previously received $\omega$-packets. Note that, when there is no information about $\bigcup_{j\in \mathcal{I}^{(i)}_{R}} \mathcal{R}^{(j)}_n$ at the transmitting node of the $i\textsuperscript{th}$ link at time $n$, it is not clear how to use the information about $\bigcup_{j\in \mathcal{I}^{(i)}_{T}} \mathcal{T}^{(j)}_n$.

To speed up the transmission of information over packet networks with feedback, CC were adopted in \cite{WL:2007} and \cite{XZWX:2008} with feedback-based scheduling policies. The idea behind such scheduling policies is that in the random scheduling policy, a transmitting node might misuse a number of transmission opportunities by transmitting some information which is not useful at the receiving node as it might be contained in previously received packets. This, therefore, increases the expected delivery time. The feedback, however, can inform the transmitting node about the set of innovative packets previously received at the receiving node, and hence the transmitting node can, in turn, avoid transmitting packets which are not innovative (with respect to the set of packets available at the receiving node) at the time of transmission.

\subsection{RP and LRF Scheduling Policies}
In \cite{WL:2007}, Wang and Li proposed a \emph{priority-based} randomized scheduling policy, referred to as \emph{random push} (RP), based on the number of innovative packets at the receiving node. In RP, for every $1\leq i\leq L$, the node transmitting over the $i\textsuperscript{th}$ link, at each time $n$, randomly chooses a chunk, say $\omega$, from the set of chunks satisfying the condition \begin{equation}\label{EQ3}\left|\bigcup_{j\in \mathcal{I}^{(\hat{i})}_{R}} \mathcal{R}^{(j)}_n(\omega)\right|>\left|\bigcup_{j\in \mathcal{I}^{(i)}_{R}} \mathcal{R}^{(j)}_n(\omega)\right|,\end{equation} where $\hat{i}$ is the label of some link whose receiving node is the transmitting node of the $i\textsuperscript{th}$ link. The transmitting node, then, generates/transmits an innovative $\omega$-packet with respect to the set $\bigcup_{j\in \mathcal{I}^{(i)}_{R}} \mathcal{R}^{(j)}_n(\omega)$, by random\footnote{The transmitting node keeps generating random linear combinations till it generates an innovative packet with respect to the set of the packets at the receiving node.} linear combination of its previously received innovative $\omega$-packets. Further, if there is no $\omega$, such that condition \eqref{EQ3} holds, the transmitting node does not transmit a packet, since, in this case, all the information available at the transmitting node is already available at the receiving node.


More recently, in \cite{XZWX:2008}, Xu \emph{et al.} introduced a deterministic scheduling policy, referred to as \emph{local-rarest-first} (LRF), by prioritizing the chunks based on the same metric as in~\cite{WL:2007}. In LRF, for every $1\leq i\leq L$, the node transmitting over the $i\textsuperscript{th}$ link, at each time $n$, selects a chunk, say $\omega$, such that (i) $\omega$ satisfies condition~\eqref{EQ3} and (ii) the size of the set $\bigcup_{j\in \mathcal{I}^{(i)}_{R}} \mathcal{R}^{(j)}_n(\omega)$ is the minimum; and generates/transmits an $\omega$-packet innovative to the set $\bigcup_{j\in \mathcal{I}^{(i)}_{R}} \mathcal{R}^{(j)}_n(\omega)$. If there exist multiple chunks satisfying both conditions (i) and (ii), one of these chunks will be selected (uniformly) at random.

\section{Proposed Scheduling Policies}
\subsection{Motivation}
The existing scheduling policies based on feedback, as discussed in Section~\ref{sec:ES}, prioritize the chunks according to the number of innovative packets at the receiving nodes at the time of transmission. In networks with delay, however, there is no guarantee that a packet which is innovative with respect to the set of packets at a receiving node at the time of transmission, would still stay innovative at the time of reception. There might be packets transmitted earlier that arrive at the receiving node at some point in time later than the time of the current transmission, but before the reception of the current transmission. Thus the set of received packets at the time of the reception of the currently transmitted packet might differ from the set at the time of the current transmission, and at that point, the currently transmitted packet might no longer be innovative.

This event particularly depends on the set of packets that are transmitted earlier but have not been received yet. The earlier a packet is transmitted, the more likely it generally is for that packet to arrive sooner, but the less likely is for that packet to deliver some useful information if it arrives.

In this work, given the link model, and the information about the set of packets transmitted by the transmitting node of a given link and the set of packets received by the receiving node of that link until a given time,\footnote{Note that if the information about the packets that were transmitted over the other links connected to the receiving node and still not received was also available at the transmitting node, a more accurate decision could be made about which chunk to choose and what packet to transmit. However, attaining such information might not be possible due to the network topology. We thus assume that such information is not available in the rest of the paper. One should also note that in the case of line networks simulated in this work, since every receiving node only receives information from one node, such situations do not apply.} we calculate the probabilities of the above mentioned events.\footnote{In earlier works~\cite{WL:2007} and\cite{XZWX:2008}, no assumption has been made about the link model, and hence such probabilities could not be calculated.} We then use these probabilities in the proposed scheduling policies. In particular, we use the expected number of innovative packets ``transmitted'' prior to the next or the current transmission time, as the metric. One should note that this is in contrast to the number of innovative packets ``received'' prior to the current transmission time, used in both~\cite{WL:2007} and~\cite{XZWX:2008}. The proposed scheduling policies are referred to as \emph{minimum-distance-first} (MDF) and \emph{minimum-current-metric-first} (MCMF), respectively.

\subsection{MDF Scheduling Policy}\label{subsec:MDF}
For every $\omega,\nu\in [q]$, let $x_{n}^{(i)}(\nu|{\omega})$ represent the expected number of innovative $\nu$-packets transmitted over the $i\textsuperscript{th}$ link prior to the next transmission time ($n+1$), given that, at the current transmission time ($n$), an innovative ${\omega}$-packet (with respect to the packets in the set $\bigcup_{j\in \mathcal{I}^{(i)}_{R}} \mathcal{R}^{(j)}_n$) is transmitted over the $i\textsuperscript{th}$ link.\footnote{Note that, in the definition of the metric $x_{n}^{(i)}(\nu|{\omega})$, the expectation is taken based on the feedback information available at time $n$.} The calculation of $x_{n}^{(i)}(\nu|\omega)$ is deferred to Section~\ref{subsec:MetricCalculation}. For every $\omega$, let $\boldsymbol{x}_{n}^{(i)}(\omega)$ represent the vector $[x_{n}^{(i)}(1|\omega),\ldots,x_{n}^{(i)}(q|\omega)]$. Let $d_{n}^{(i)}(\omega)$ denote the Euclidean distance between the vector $\boldsymbol{x}_{n}^{(i)}(\omega)$ and the ($q$-dimensional) vector $[k/q,\ldots,k/q]$.

In MDF, the node transmitting over the $i\textsuperscript{th}$ link, at each time $n$, selects the chunk $\omega$ such that (i) $\omega$ satisfies condition~\eqref{EQ3} and (ii) $d_{n}^{(i)}(\omega)$ is minimized. That is, the transmitting node chooses a chunk whose transmission at the present time minimizes the distance between the vector of the ``expected'' number of innovative packets transmitted (over the $i\textsuperscript{th}$ link) prior to the next transmission time and the vector $[k/q,\ldots, k/q]$. Note that reaching the latter vector is the goal of the network coding solution (i.e., all the chunks can be successfully decoded so long as there are $k/q$ innovative packets pertaining to each chunk). Therefore, the MDF scheduling policy is devised to achieve this goal in a greedy fashion by taking the largest possible step towards (by obtaining the smallest distance from) the target. Despite the fact that the MDF scheduling policy is heuristic, in Section~\ref{subsec:OptimalityMDF}, we present some experimental results that indicate the (near) optimality of this scheme over line networks (where the source node and the sink node are connected through the internal nodes connected in tandem) in the sense of minimizing the expected delivery time.\footnote{It should be noted that, currently, no analytical result on the proposed or the existing scheduling policies, for a given link model, is available. The difficulty of such analysis stems from the high-level of dependency between the large number of random variables involved in the process.} Similar to LRF and RP, in MDF, if chunk $\omega$ is chosen, the transmitting node randomly generates/transmits an $\omega$-packet innovative to the packets at the receiving node.

\subsection{MCMF Scheduling Policy}\label{subsec:MCMF}
For every $\nu\in [q]$, let $y^{(i)}_n(\nu)$ represent the expected number of innovative $\nu$-packets transmitted over the $i\textsuperscript{th}$ link prior to the current transmission time $n$. It should be clear that, by the definition, $y^{(i)}_n(\nu)=x^{(i)}_n(\nu|\omega)$, for any $\omega\neq\nu$.\footnote{Both the metrics $y^{(i)}_n(\nu)$ and $x^{(i)}_n(\nu|\omega)$, i.e., the expected number of innovative packet transmissions prior to the current and the next transmission time, respectively, pertaining to any chunk $\nu$ ($\nu\neq\omega$), are equal since they both rely on the same (feedback) information till the current transmission time $n$.} In MCMF, the node transmitting over the $i\textsuperscript{th}$ link, at each time $n$, selects the chunk $\omega$ such that (i) $\omega$ satisfies condition~\eqref{EQ3} and (ii) $y^{(i)}_{n}(\omega)$ is minimized.

\begin{lemma}\label{lem:1}For networks with unit-delay links (defined in Section~\ref{subsec:ScheduleModel}), MDF policy reduces to MCMF policy.\end{lemma}

\begin{proof}Since the delay values are all one, at the current transmission time, there is no randomness in the number of innovative packet transmissions pertaining to any chunk prior to this time. Thus, by transmitting a given chunk at the current transmission time, the expected number of innovative packet transmissions (prior to the next transmission time) pertaining to that chunk increases, yet, this number does not change for the rest of the chunks. The amount of this increase by transmitting any chunk is the same as that by transmitting any other chunk, and hence, in such a case, MDF reduces to MCMF, which operates by choosing a chunk which has the smallest (expected) number of innovative packets transmitted prior to the current transmission.\end{proof}

For a network with a general delay model, MDF outperforms MCMF in terms of the expected delivery time. The performance advantage is more profound for random delays with larger mean and variance, for larger networks and for smaller chunks. The performance improvement for MDF policy however, comes at the expense of higher computational complexity.

\subsection{Metric Calculation}\label{subsec:MetricCalculation}
For every $\omega,\nu\in [q]$, we need to calculate the metric $x^{(i)}_n(\nu|\omega)$ for the MDF policy. Note that, by the definition, in order to calculate $x^{(i)}_n(\nu|\omega)$, we focus on the sets $\bigcup_{j\in \mathcal{I}^{(i)}_{R}} \mathcal{R}^{(j)}_n(\nu)$ and $\mathcal{T}^{(i)}_n(\nu)$, and assume that, at the current time $n$, an innovative $\omega$-packet with respect to the set $\bigcup_{j\in \mathcal{I}^{(i)}_{R}} \mathcal{R}^{(j)}_n(\omega)$ is transmitted over the $i\textsuperscript{th}$ link.


For every chunk $\nu$, every link $i$, and every time $n$, let $\rho^{(i)}_n(\nu)=|\bigcup_{j\in \mathcal{I}^{(i)}_{R}} \mathcal{R}^{(j)}_n(\nu)|$ and $\tau^{(i)}_n(\nu)=|\mathcal{T}^{(i)}_n(\nu)\setminus\bigcup_{j\in \mathcal{I}^{(i)}_{R}} \mathcal{R}^{(j)}_n(\nu)|$. Furthermore, let $\mathcal{U}^{(i)}_n(\nu)$ denote the set $\bigcup_{j\in \mathcal{I}^{(i)}_{R}} \mathcal{R}^{(j)}_n(\nu)$. For the ease of notation, hereafter, we often drop the argumant $\nu$, the subscript $n$, and superscript $i$, unless there is a possibility for confusion. For example, we use the notations $\rho$ and $\tau$, instead of $\rho^{(i)}_n(\nu)$ and $\tau^{(i)}_n(\nu)$, respectively.

Let $\mathcal{N}_{r}$ and $\mathcal{N}_{t}$ be the set of the time indices that the $\nu$-packets in $\mathcal{U}$ and $\mathcal{T}\setminus\mathcal{U}$ are received and transmitted, respectively, in an increasing order, i.e., $\mathcal{N}_{r}=\{r_1,\dots,r_{\rho}\}$, and $\mathcal{N}_{t}=\{t_1,\dots,t_{\tau}\}$, so that $r_1\leq \cdots \leq r_{\rho}$, and $t_1\leq \cdots \leq t_{\tau}$. To lower the computational complexity of the scheduling policy, for some constant integer $m\leq \tau$, we focus on the set of $m$ packets in $\mathcal{T}\setminus\mathcal{U}$, transmitted at the time indices $\mathcal{N}_{t,m}=\{t_{\tau-m+1},\ldots,t_{\tau}\}$, i.e., the last $m$ packets transmitted but not received up to time $n$. Taking into account only $m$ out of $\tau$ delayed packets, however, results in an approximation of $x_n^{(i)}(\nu|\omega)$.\footnote{The smaller is the value of $m$, the lower is the complexity of the scheduling policy (and the smaller is the memory requirement at the network nodes). This is at the expense of larger approximation error.}

Let $\tau^{*}=\tau-m+1$. For the case of $\omega=\nu$, we define $\boldsymbol{z}=\{z_{t_{\tau^{*}}},\dots,z_{t_{\tau}},z_{n}\}$ as the sequence of the delays that the packets transmitted at time indices $\{\mathcal{N}_{t,m},n\}$ experience, assuming that all these $m+1$ packets arrive (the last packet is the one that, we assume, is transmitted at the time $n$), i.e., for every $\tau^{*}\leq j\leq \tau$, the packet transmitted at time $t_j$ arrives at time $t_j+z_{t_j}$, and the packet transmitted at time $n$ arrives at time $n+z_n$. For the reason that none of these packets has been received till time $n$, for every $j$, the delay $z_{t_j}$ is bounded from below by $n-t_j$, for every possible delay sequence $\boldsymbol{z}$. For the other cases of $\omega\neq\nu$, due to the fact that the packet which is assumed to be transmitted at time $n$ is an $\omega$-packet, the sequence $\boldsymbol{z}$ excludes the term $z_n$. In such cases, we denote the truncated sequence $\boldsymbol{z}$ by $\boldsymbol{z}_T$.

The delays are, however, random variables that can sometimes take very large values, and it is thus not practical to consider the set of all possible delay sequences $\boldsymbol{z}$. To lower the computational complexity of the scheduling policy, we introduce a constant integer $\Delta$, so that if a packet transmitted prior to time $n$ (or transmitted at time $n$) is assumed to arrive later than ${\Delta}$ time units after the time $n$ (or $n+1$), it will be treated as an erased packet in our calculations. We, thus, focus on a subset of all possible delay sequences, referred to as the \emph{desirable sequences}, so that at time $n$, for every $\omega\neq \nu$, the delay of the $j\textsuperscript{th}$ packet ($\tau^{*}\leq j\leq \tau$) is bounded above by $n-t_j+\Delta$, and for $\omega=\nu$, the delay of the last packet (assumed to be transmitted at time $n$) is bounded above by $\Delta$. For the desirable sequences, we thus have: for every $\tau^{*}\leq j\leq \tau$, $n-t_j<z_{t_j}\leq n-t_j+\Delta$, and $0<z_{n}\leq\Delta$.\footnote{The smaller is the choice of $\Delta$, the smaller is the number of desirable delay sequences to be taken into account and hence the lower is the computational complexity of the scheduling policy. This however comes at the expense of larger approximation error.}

For the sake of brevity, hereafter, we focus on the case with $\omega=\nu$. Clearly, by removing the terms related to the packet transmission at time $n$ and its delay value $z_n$, the other cases with $\omega\neq\nu$ will be covered.\footnote{One should note that, for a fixed chunk $\nu$, the metrics $x^{(i)}_n(\nu|\omega)$, for all $\omega\neq\nu$, are the same (independent of $\omega$), and hence need to be calculated only once.}

Let $t_{\tau+1}\triangleq n$. For every desirable $\boldsymbol{z}$, suppose that its elements are reordered as follows: let the sequence $\{t'_{\tau^{*}}+z_{t'_{\tau^{*}}},\dots,t'_{\tau+1}+z_{t'_{\tau+1}}\}$ represent the sequence $\{t_{\tau^{*}}+z_{t_{\tau^{*}}},\dots,t_{\tau+1}+z_{t_{\tau+1}}\}$ sorted in an increasing order, i.e., $t'_{\tau^{*}}+z_{t'_{\tau^{*}}}\leq \dots \leq t'_{\tau+1}+z_{t'_{\tau+1}}$, and for every $\tau^{*}\leq i\leq \tau+1$, there exists a unique $\tau^{*}\leq j\leq \tau+1$, such that $t_i=t'_j$. For every sequence $\boldsymbol{z}$, hereafter, we use its corresponding reordered sequence based on the reception time indices, and adopt the same notation $\boldsymbol{z}$ to represent it.

For every desirable $\boldsymbol{z}$, the probability that a packet, which is transmitted over the $i\textsuperscript{th}$ link at time $t'_j$, but not received till time $n$, arrives after a delay $n-t'_j<z_{t'_j}\leq n-t'_j+\Delta$, for every $\tau^{*}\leq j\leq \tau$, is \begin{equation}\label{EQ4}p[z_{t'_j}]=\frac{P_{Z^{(i)}_{n}}[z_{t'_j}]}{1-\sum_{1\leq z\leq n-t'_j}P_{Z^{(i)}_{n}}[z]}\cdot\left(1-p^{(i)}_e\right),\end{equation} and \begin{equation}\label{EQ5}p[z_{t'_{\tau+1}}]=P_{Z^{(i)}_{n}}[z_{t'_{\tau+1}}]\cdot\left(1-p^{(i)}_e\right),\end{equation} where $P_{Z^{(i)}_{n}}$ is given by~\eqref{EQ2}, and $p^{(i)}_e$ is the probability of erasure over the $i\textsuperscript{th}$ link.

The packets which will (will not) arrive at the receiving node till the next $\Delta$ time units are referred to as \emph{on-time} (\emph{late}). One should note that some late packets might be erased (not be successful) and will never arrive at the receiving node. By the definition, however, all the on-time packets are successful. It should be clear that some of the $m+1$ packets might not be on-time, and the on-time packets might not arrive in the same order that they were transmitted (any possible subset of the $m+1$ packets might be on-time with any possible ordering). The innovation of a packet at the time of reception, however, is dependent on the set of packets that arrived earlier along with the order in which they arrive. We thus need to differentiate between the two partitions of on-time and late packets.

For every possible subset of on-time packets, let us consider a binary sequence (of $m+1$ elements) $\boldsymbol{b}=\{b_{t'_{\tau^{*}}},\dots,b_{t'_{\tau+1}}\}$, such that, for all $\tau^{*}\leq j\leq \tau+1$, $b_{t'_j}$ is $1$, if the packet transmitted at the time $t'_j$ is assumed to be on-time, and $b_j$ is $0$, otherwise. In particular, for every $\boldsymbol{z}=\{z_{t'_{\tau^{*}}},\dots,z_{t'_{\tau+1}}\}$, the packet transmitted at the time $t'_j$ is assumed to be on-time and to experience a delay $z_{t'_j}$, if $b_{t'_j}$ is $1$, and the packet will be late (i.e., either is successful but does not arrive on-time, or is not successful and is erased), if $b_{t'_j}$ is $0$. Thus the (joint) probability that all the packets whose corresponding binary elements are $1$ arrive on-time with their corresponding delays, and that the rest of the packets are late (regardless of their corresponding delay values), is \begin{dmath*}p_{_{\boldsymbol{b},\boldsymbol{z}}}=\prod_{\tau^{*}\leq j\leq \tau+1}\left(b_{t'_j} p[z_{t'_j}]+(1-b_{t'_j})\left(1-\sum_{n-t'_j<z_j\leq n-t'_j+\Delta}p[z_j] + p^{(i)}_e\right)\right),\end{dmath*} for every $\boldsymbol{b}\in \{0,1\}^{m+1}$, and every desirable sequence $\boldsymbol{z}$, where $p[z_{t'_j}]$ is given in~\eqref{EQ4} and~\eqref{EQ5}, for every $\tau^{*}\leq j\leq \tau$, and $j=\tau+1$, respectively. (For the cases of $\omega\neq \nu$, the sequence $\boldsymbol{b}$ excludes the term $b_{m+1}$, and we denote the truncated sequence $\boldsymbol{b}$ with $\boldsymbol{b}_T\in\{0,1\}^{m}$.)



For every $\boldsymbol{b}$, let $m^{*}$ denote the number of $1$'s in $\boldsymbol{b}$. Now, consider the subset $\{z_{{i_1}},\dots,z_{{i_{m^{*}}}}\}$ of the elements of $\boldsymbol{z}$ whose corresponding elements in the sequence $\boldsymbol{b}$ are $1$. Correspondingly, let $\{i_1,\dots,i_{m^{*}}\}$ denote the associated sequence of the transmission time indices. For every $1\leq \ell\leq m^{*}$, let us define $\mathcal{N}_{{\boldsymbol{b},\boldsymbol{z}}|\ell}$ as the subset of all the reception time indices $\{i_1+z_{i_1},\dots,i_{\ell}+z_{i_{\ell}}\}$ whose corresponding packets are innovative to the set of packets with the reception time indices $\mathcal{N}_{\boldsymbol{b},\boldsymbol{z}|\ell-1}\cup \mathcal{N}_{r}$. Note that, $\mathcal{N}_{\boldsymbol{b},\boldsymbol{z}|0}$ is the empty set.

To indicate whether the $\ell\textsuperscript{th}$ packet is innovative at the time of reception, we introduce an indicator variable $I_{\boldsymbol{b},\boldsymbol{z}|\ell}$ defined as follows: $I_{\boldsymbol{b},\boldsymbol{z}|\ell}$ is $1$, if the packet with the reception time $i_{\ell}+z_{i_{\ell}}$ is innovative to the set of packets with the reception time indices $\mathcal{N}_{\boldsymbol{b},\boldsymbol{z}|\ell-1}\cup \mathcal{N}_{r}$, and $I_{\boldsymbol{b},\boldsymbol{z}|\ell}$ is $0$, otherwise. Thus, for every $\nu$, at time $n$, the expected number of innovative $\nu$-packets transmitted over the $i\textsuperscript{th}$ link prior to the next transmission time, given that an innovative $\nu$-packet is transmitted over the $i\textsuperscript{th}$ link at time $n$, can be calculated as \begin{equation}\label{EQ6}x_n^{(i)}(\nu|\nu)=\sum_{\boldsymbol{b}}\sum_{\boldsymbol{z}} p_{_{\boldsymbol{b},\boldsymbol{z}}}\cdot \left(\rho+\sum_{1\leq \ell\leq m^{*}}I_{\boldsymbol{b},\boldsymbol{z}|\ell}\right).\end{equation} Similarly, for every $\omega\neq\nu$, $x_n^{(i)}(\nu|\omega)$ can be calculated by~\eqref{EQ6}, where $\boldsymbol{b}$ and $\boldsymbol{z}$ are replaced with $\boldsymbol{b}_T$ and $\boldsymbol{z}_T$, respectively, i.e., \begin{equation}\label{EQ7}x_n^{(i)}(\nu|\omega)=\sum_{\boldsymbol{b}_T}\sum_{\boldsymbol{z}_T} p_{_{\boldsymbol{b}_T,\boldsymbol{z}_T}}\cdot \left(\rho+\sum_{1\leq \ell\leq m^{*}_T}I_{\boldsymbol{b}_T,\boldsymbol{z}_T|\ell}\right),\end{equation} where $m^{*}_T$ denotes the number of $1$'s in $\boldsymbol{b}_T$. Note that, since $y^{(i)}_n(\nu)=x^{(i)}_{n}(\nu|\omega)$ for every $\omega$ ($\neq\nu$), the metric $y_n^{(i)}(\nu)$ for MCMF can also be calculated by~\eqref{EQ7}.

\begin{table*}[t]
  \centering
  \caption{Parameters of Delay Models Used in the Simulations}
    \begin{tabular}{|cr|c|c|c|r|r|}
    \hline
    \multicolumn{2}{|c|}{Network Length} & \multicolumn{5}{c|}{$L$} \\
    \hline
    \multicolumn{2}{|c|}{Delay Model} & I & II & III & \multicolumn{1}{c|}{IV} & \multicolumn{1}{c|}{V} \\
    \hline\hline
    \multicolumn{2}{|c|}{$(\mu_i,\sigma_i)$} & $(0.5,0.5)$ & $(1,0.5)$ & $(1,1)$ & \multicolumn{1}{c|}{$\begin{array}{ll}(0.5,0.5), & \text{if } 1\leq i\leq {L}/{2}; \\ (1,1), & \text{otherwise.} \end{array}$} & \multicolumn{1}{c|}{$\begin{array}{ll}(1,1), & \text{if } 1\leq i\leq {L}/{2}; \\ (0.5,0.5), & \text{otherwise.} \end{array}$} \\
    \hline
    \multicolumn{2}{|c|}{$(\text{E}(R^{(i)}),\text{Var}(R^{(i)}))$} & $(1.86,0.99)$ & $(3.08,2.69)$ & $(4.48,34.51)$ & \multicolumn{1}{c|}{$\begin{array}{ll}(1.86,0.99), & \text{if } 1\leq i\leq {L}/{2}; \\ (4.48,34.51), & \text{otherwise.} \end{array}$} & \multicolumn{1}{c|}{$\begin{array}{ll}(4.48,34.51), & \text{if } 1\leq i\leq {L}/{2}; \\ (1.86,0.99), & \text{otherwise.} \end{array}$} \\
    \hline
    \end{tabular}%
  \label{tab:DelayModelsPars}%
\end{table*}%

\subsection{On the Amount of Feedback and the Computational Complexity of the Proposed Scheduling Policies}
It is worth noting that both MDF and MCMF require more feedback and more computations compared to RP and LRF. Part of the feedback in MDF and MCMF is required to transmit the link parameters, estimated at the receiver, to the transmitter. This however is not needed for RP and LRF policies. Moreover, in RP and LRF policies, the transmitting node only requires the set of innovative packets at the receiving node. It however does not require the departure/arrival time of such packets. Such information, on the other hand, is required to calculate the metrics in MDF and MCMF policies.

Unlike RP and LRF policies, MDF and MCMF policies need to estimate the link parameters (for erasure and delay). This increases the computational complexity and transmission overhead of the proposed policies.\footnote{Efficient techniques for link estimation can be found in~\cite{Paxson:1997,Moon:2000}, and are beyond the scope of this paper.} The main part of the computational complexity of MDF and MCMF policies is however dedicated to the calculation of $x_n^{(i)}(\nu|\omega)$, for every $\omega,\nu\in[q]$. This complexity corresponds to the calculation of the double-summations in~\eqref{EQ6} and/or~\eqref{EQ7} over all the desirable delay sequences $\boldsymbol{z}$ and/or $\boldsymbol{z}_T$, and the binary sequences $\boldsymbol{b}$ and/or $\boldsymbol{b}_T$. Part of the computations of the argument of the double-summations can be carried out offline, and the results can be stored for online use. (This part is the calculation of the values of $p_{_{\boldsymbol{b},\boldsymbol{z}}}$ and $p_{_{\boldsymbol{b}_T,\boldsymbol{z}_T}}$.) The values of $I_{\boldsymbol{b},\boldsymbol{z}|\ell}$ and $I_{\boldsymbol{b}_T,\boldsymbol{z}_T|\ell}$ however need to be computed online, as they depend on the actual set of innovative received packets. To determine each value of $I_{\boldsymbol{b},\boldsymbol{z}|\ell}$ or $I_{\boldsymbol{b}_T,\boldsymbol{z}_T|\ell}$, one needs to find the rank of a matrix formed by the global encoding vectors of the packets under consideration. It should however be noted that such operations are performed in the field associated with the linear coding scheme, and are in general negligible in comparison with packet operations required for encoding, particularly for larger packet sizes (see, e.g., \cite{MHL:2006}). In addition, for coding schemes operating over finite fields of large size, the summations $\sum_{1\leq \ell\leq m^{*}}I_{\boldsymbol{b},\boldsymbol{z}|\ell}$ and $\sum_{1\leq \ell\leq m^{*}_T}I_{\boldsymbol{b}_T,\boldsymbol{z}_T|\ell}$ in~\eqref{EQ6} and~\eqref{EQ7} can be simply approximated based on the number and the ordering of the on-time packets depending on the sequences $\boldsymbol{b}$ and $\boldsymbol{z}$, or $\boldsymbol{b}_T$ and $\boldsymbol{z}_T$, respectively.

Finally, as mentioned earlier, the computational complexity of MCMF is smaller than that of MDF. This arises from the following facts: (i) in MCMF, for every link $i$ and every time instant $n$, the metric $y^{(i)}_n(\nu)$ (which is equal to $x^{(i)}_{n}(\nu|\omega)$, for any $\omega\neq\nu$), for every chunk $\nu$, needs to be calculated. However, in MDF, the two metrics $x^{(i)}_{n}(\nu|\nu)$ and $x^{(i)}_{n}(\nu|\omega)$, for some $\omega\neq\nu$ need to be calculated. This implies that the computational complexity of MDF is at least twice the computational complexity of MCMF; (ii) in MDF, in order to calculate the metric $x^{(i)}_{n}(\nu|\nu)$, the sequences $\boldsymbol{b}$ and $\boldsymbol{z}$ are each of length $m+1$. However, in MCMF, in order to calculate the metric $y^{(i)}_n(\nu)=x^{(i)}_{n}(\nu|\omega)$, for some $\omega\neq\nu$, the sequences $\boldsymbol{b}_T$ and $\boldsymbol{z}_T$ are each of length $m$, and hence the calculation of the double-summation in~\eqref{EQ7} requires less computations compared to that in~\eqref{EQ6}; and (iii) In MDF, having the metric vectors $\boldsymbol{x}^{(i)}_n(\omega)$, for all chunks $\omega$, the Euclidian distances $d^{(i)}_n(\omega)$ need to be calculated, and then the chunk with minimum distance will be chosen. However, in MCMF, having the metrics $y^{(i)}_n(\nu)$, for all chunks $\nu$, the chunk with minimum metric will be chosen, and there is no need for further computation.
\begin{table*}[t]
  \centering
  \caption{Expected Delivery Time for Various Scheduling Policies over Identical Lossless Links with Delay}
    \begin{tabular}{|c|r|c|c|c|c|c|c|c|c|c|c|c|c|c|}
    \hline
    \multirow{11}[22]{*}{\vspace{1.5cm}\begin{sideways}Lossless\end{sideways}} & \multicolumn{2}{c|}{Delay Model} & \multicolumn{4}{c|}{I} & \multicolumn{4}{c|}{II} & \multicolumn{4}{c|}{III} \\
\cline{2-15}      & \multicolumn{2}{c|}{Network Length} & \multicolumn{2}{c|}{$2$} & \multicolumn{2}{c|}{$8$} & \multicolumn{2}{c|}{$2$} & \multicolumn{2}{c|}{$8$} & \multicolumn{2}{c|}{$2$} & \multicolumn{2}{c|}{$8$} \\
\cline{2-15}      & \multicolumn{2}{c|}{Chunk Size} & $8$ & $32$ & $8$ & $32$ & $8$ & $32$ & $8$ & $32$ & $8$ & $32$ & $8$ & $32$ \\
\cline{2-15}      & \multicolumn{1}{r|}{\multirow{5}[10]{*}{\vspace{0.5cm}\begin{sideways}{\fontsize{4.5}{4.5}\selectfont Scheduling Policy}\end{sideways}}} & Random & $156.45$ & $89.46$ & $331.78$ & $150.56$ & $162.80$ & $91.66$ & $345.86$ & $158.49$ & $167.66$ & $94.14$ & $351.62$ & $168.38$ \\
\cline{3-15}      &   & RP & $102.12$ & $81.82$ & $170.23$ & $135.08$ & $106.01$ & $85.19$ & $191.57$ & $149.30$ & $111.64$ & $88.59$ & $199.53$ & $155.91$ \\
\cline{3-15}      &   & LRF & $102.00$ & $79.81$ & $182.16$ & $130.21$ & $107.71$ & $83.63$ & $205.81$ & $143.21$ & $111.82$ & $86.15$ & $215.78$ & $151.56$ \\
\cline{3-15}      &   & MDF & $69.52$ & $70.77$ & $91.81$ & $96.26$ & $73.02$ & $76.07$ & $103.95$ & $111.75$ & $82.20$ & $86.05$ & $130.26$ & $130.64$ \\
\cline{3-15}      &   & MCMF & $76.41$ & $76.19$ & $111.12$ & $104.59$ & $91.15$ & $81.33$ & $142.42$ & $124.53$ & $107.73$ & $86.10$ & $153.48$ & $131.85$ \\
\cline{2-15}      & \multicolumn{2}{c|}{$I_1$ ($\%$)} & $31.84$ & $11.33$ & $46.07$ & $26.07$ & $31.12$ & $9.04$ & $45.74$ & $13.04$ & $26.37$ & $0.10$ & $34.72$ & $13.80$ \\
\cline{2-15}      & \multicolumn{2}{c|}{$I_2$ ($\%$)} & $25.08$ & $4.53$ & $34.72$ & $19.67$ & $14.01$ & $2.75$ & $25.65$ & $13.04$ & $3.50$ & $0.05$ & $23.07$ & $13.00$ \\
    \hline
    \end{tabular}%
  \label{tab:ExpSuccessTimeIdentical}%
\end{table*}%

\begin{table*}[t]
  \centering
  \caption{Expected Delivery Time for Various Scheduling Policies over Non-Identical Lossless Links with Delay}
    \begin{tabular}{|c|r|c|c|c|c|c|c|c|c|c|}
    \hline
    \multirow{11}[22]{*}{\vspace{1.5cm}\begin{sideways}Lossless\end{sideways}} & \multicolumn{2}{c|}{Delay Model} & \multicolumn{4}{c|}{IV} & \multicolumn{4}{c|}{V} \\
\cline{2-11}      & \multicolumn{2}{c|}{Network Length} & \multicolumn{2}{c|}{$2$} & \multicolumn{2}{c|}{$8$} & \multicolumn{2}{c|}{$2$} & \multicolumn{2}{c|}{$8$} \\
\cline{2-11}      & \multicolumn{2}{c|}{Chunk Size} & $8$ & $32$ & $8$ & $32$ & $8$ & $32$ & $8$ & $32$ \\
\cline{2-11}      & \multicolumn{1}{r|}{\multirow{5}[10]{*}{\vspace{0.5cm}\begin{sideways}{\fontsize{4.5}{4.5}\selectfont Scheduling Policy}\end{sideways}}} & Random & $161.80$ & $91.89$ & $340.62$ & $159.40$ & $162.38$ & $91.71$ & $341.56$ & $159.45$ \\
\cline{3-11}      &   & RP & $107.39$ & $86.00$ & $187.55$ & $145.37$ & $103.49$ & $84.84$ & $182.85$ & $144.87$ \\
\cline{3-11}      &   & LRF & $107.12$ & $83.69$ & $205.51$ & $141.70$ & $105.21$ & $83.38$ & $194.70$ & $140.80$ \\
\cline{3-11}      &   & MDF & $76.42$ & $79.83$ & $117.43$ & $118.78$ & $73.85$ & $77.04$ & $112.37$ & $117.80$ \\
\cline{3-11}      &   & MCMF & $86.21$ & $80.77$ & $148.91$ & $129.88$ & $94.59$ & $78.98$ & $137.45$ & $119.52$ \\
\cline{2-11}      & \multicolumn{2}{c|}{$I_1$ ($\%$)} & $28.65$ & $4.61$ & $37.38$ & $16.17$ & $28.64$ & $7.60$ & $38.54$ & $16.33$ \\
\cline{2-11}      & \multicolumn{2}{c|}{$I_2$ ($\%$)} & $19.52$ & $3.48$ & $20.60$ & $8.34$ & $8.59$ & $5.27$ & $24.82$ & $15.11$ \\
    \hline
    \end{tabular}%
  \label{tab:ExpSuccessTimeNonIdentical}%
\end{table*}%

\section{Simulation Results}\label{sec:SR}
We compare random, RP, LRF and MDF scheduling policies over line networks with one source node, one sink node and $L-1$ internal nodes connected in tandem. The comparisons are in terms of the expected delivery time (i.e., the expected time it takes for all the chunks to be decodable). The variables involved in the comparisons are the size of the chunks, the length of the network and the parameters of the delay and the loss models. We present the simulation results in two parts: lossless links with (random) delays, and lossy links with unit delays. By combining these results, one can easily generalize the results to the case of the links with both loss and delay. In each part, we consider two cases: identical links and non-identical links.

\subsection{Lossless Links with Delay}\label{subsec:LosslessLinksWithDelay}
We consider line networks of lengths $2$ and $8$ (i.e., $L\in\{2,8\}$). The links are assumed to be lossless, and for every $1\leq i\leq L$, the delay model of the $i\textsuperscript{th}$ link is specified as follows: The (continuous delay) probability distribution $f_{R^{(i)}}(r)$, used in~\eqref{EQ2}, is assumed to be log-normal\footnote{It has been recently shown that, in a variety of real-world packet networks, the delay can be modeled by a heavy-tailed distribution (i.e., the right, or left, or both tail(s) of the probability distribution function are not exponentially bounded), see, e.g., \cite{T:2009}. Examples of such distributions are log-normal, Pareto, and L\'{e}vy.} with the location and scale parameters $(\mu_i,\sigma_i)$, i.e., \[f_{R^{(i)}}(r)=\frac{1}{r \sigma_i \sqrt{2 \pi}}e^{-\frac{(\ln r - \mu_i)^2}{2\sigma^2_i}},\] where $\{(\mu_i,\sigma_i)\}$ are specified in Table~\ref{tab:DelayModelsPars}. The mean and the variance of a log-normal random variable with the location and scale parameters $(\mu,\sigma)$ are $e^{\mu+\sigma^2/2}$ and $(e^{\sigma^2}-1)e^{2\mu+\sigma^2}$, respectively. In the case of identical links, we consider three delay models, labeled as delay models \textmd{I}, \textmd{II} and \textmd{III}; and in the case of non-identical links, we consider two delay models, labeled as delay models \textmd{IV} and \textmd{V}.

The message size is assumed to be $64$. We consider two sizes of chunks: $8$ and $32$. For each set of chunk size, delay model and network length, $100$ network realizations are simulated, and the chunked coding scheme (over the binary field) along with each scheduling policy is applied to each network realization for $100$ trials. For the MDF and MCMF scheduling policies, the parameters $m$ and ${\Delta}$ are set to $4$ and $4$, respectively. It should be noted that the selected values of $m$ and $\Delta$ strike a good balance between the complexity of simulations and the accuracy of the results for the purpose of the comparisons in this paper. To determine the expected delivery time, for each case, the expectation is taken over the $100$ network realizations and the $100$ trials of the coding scheme.

Tables~\ref{tab:ExpSuccessTimeIdentical} and~\ref{tab:ExpSuccessTimeNonIdentical} list the expected delivery time for each scheduling policy in the case of identical and non-identical lossless links with delay, respectively. Each table quickly reveals that all the scheduling policies significantly outperform the random scheduling policy. The last two rows of each table present the relative performance of the proposed scheduling policies compared to the existing feedback-based scheduling policies. Parameters $I_1$ and $I_2$ are defined as $I_1 = \frac{\min\{RP,LRF\}-MDF}{\min\{RP,LRF\}}$ and $I_2=\frac{\min\{RP,LRF\}-MCMF}{\min\{RP,LRF\}}$, respectively, where, e.g., $LRF$ denotes the expected delivery time of the LRF scheduling policy.

As it can be seen in Table~\ref{tab:ExpSuccessTimeIdentical}, both MDF and MCMF policies outperform RP and LRF policies. The largest improvement in this table is $46.07\%$ and corresponds to the MDF scheduling policy over a network of length $8$ with delay model \textmd{I} where the chunk size is $8$. The improvement of MDF/MCMF policies over RP/LRF policies is larger for delays with smaller mean and variance. For example, considering the MDF scheduling policy, for the case of the chunk size $8$ and the network length $8$, it can be observed that $I_1=46.07\%$ for the delay model \textmd{I}. It is then reduced to $I_1=45.74\%$, for the delay model \textmd{II} (with larger mean and variance), and is further reduced to $34.72\%$ for the delay model \textmd{III}. Furthermore, the advantage of MDF/MCMF over RP/LRF becomes more for smaller chunks and larger networks. For example, in the case of MDF over the delay model \textmd{I} and the network length $8$, for the (larger) chunk size $32$, one can see that $I_1=26.07\%$, which is smaller than that for the (smaller) chunk size $8$ (i.e., $I_1=45.74\%$); or in the case of MDF over the delay model \textmd{I} and the chunk size $8$, for the (smaller) network length $2$, it can be seen that $I_1=31.84\%$, which is smaller than that for the (larger) network length $8$ (i.e., $I_1=46.07\%$). Similar trends can also be observed for the MCMF scheduling policy. Furthermore, comparing the advantages of MDF and MCMF over RP/LRF (by comparing the values of $I_1$ and $I_2$), it can be easily seen that MDF always outperforms MCMF (i.e., for each case, $I_1\geq I_2$).

Similarly, in Table~\ref{tab:ExpSuccessTimeNonIdentical}, for the case of non-identical links, one can observe similar trends as in the case of identical links, for a given delay model, i.e., the advantage of MDF/MCMF over RP/LRF is more pronounced for smaller chunks and larger networks.

Based on the results in Tables~\ref{tab:ExpSuccessTimeIdentical} and~\ref{tab:ExpSuccessTimeNonIdentical}, the relative performance of LRF and RP (or MDF and MCMF) compared to each other and compared to the random scheduling policy, are listed in Tables~\ref{tab:RelativeToRandomLosslessDelayIdentical} and~\ref{tab:RelativeToRandomLosslessDelayNonIdentical}. For each scheduling policy, e.g., LRF, $I_R$ is defined as $I_R=\frac{R-LRF}{R}$, where $R$ denotes the expected delivery time of the random scheduling policy. For the pair of scheduling policies RP and LRF (or MDF and MCMF), the parameter $I_E$ (or $I_P$) is defined as $I_E = \frac{LRF-RP}{LRF}$ (or $I_P=\frac{MCMF-MDF}{MCMF}$).

\begin{table*}[t]
  \centering
  \caption{Relative Performance of Scheduling Policies over Identical Lossless Links with Delay}
    \begin{tabular}{|r|c|c|c|c|c|c|c|c|c|c|c|c|c|c|c|}
    \hline
    \multicolumn{1}{|r|}{\multirow{9}[18]{*}{\vspace{1cm}\begin{sideways}Lossless\end{sideways}}} & \multicolumn{3}{c|}{Delay Model} & \multicolumn{4}{c|}{I} & \multicolumn{4}{c|}{II} & \multicolumn{4}{c|}{III} \\
\cline{2-16}    \multicolumn{1}{|r|}{} & \multicolumn{3}{c|}{Network Length} & \multicolumn{2}{c|}{$2$} & \multicolumn{2}{c|}{$8$} & \multicolumn{2}{c|}{$2$} & \multicolumn{2}{c|}{$8$} & \multicolumn{2}{c|}{$2$} & \multicolumn{2}{c|}{$8$} \\
\cline{2-16}    \multicolumn{1}{|r|}{} & \multicolumn{3}{c|}{Chunk Size} & $8$ & $32$ & $8$ & $32$ & $8$ & $32$ & $8$ & $32$ & $8$ & $32$ & $8$ & $32$ \\
\cline{2-16}    \multicolumn{1}{|r|}{} & \multirow{4}[8]{*}{\vspace{0.35cm}\begin{sideways}{\fontsize{4.5}{4.5}\selectfont Scheduling Policy}\end{sideways}} & RP & \multirow{4}[8]{*}{\vspace{0.45cm}$I_R$ ($\%$)} & $34.72$ & $8.54$ & $48.69$ & $10.28$ & $34.88$ & $7.05$ & $44.61$ & $5.79$ & $33.41$ & $5.89$ & $43.25$ & $7.40$ \\
\cline{3-3}\cline{5-16}    \multicolumn{1}{|r|}{} &   & LRF &   & $34.80$ & $10.78$ & $45.09$ & $13.51$ & $33.83$ & $8.76$ & $40.49$ & $9.64$ & $33.30$ & $8.48$ & $38.63$ & $9.98$ \\
\cline{3-3}\cline{5-16}    \multicolumn{1}{|r|}{} &   & MDF &   & $55.56$ & $20.89$ & $72.32$ & $36.06$ & $55.14$ & $17.00$ & $69.94$ & $29.49$ & $50.97$ & $8.59$ & $62.95$ & $22.41$ \\
\cline{3-3}\cline{5-16}    \multicolumn{1}{|r|}{} &   & MCMF &   & $51.16$ & $14.83$ & $66.50$ & $30.53$ & $44.01$ & $11.26$ & $58.82$ & $21.42$ & $35.74$ & $8.54$ & $56.35$ & $21.69$ \\
\cline{2-16}      & \multicolumn{3}{c|}{$I_{E}$ ($\%$)} & $-0.11$ & $-2.51$ & $+6.54$ & $-3.74$ & $+1.57$ & $-1.86$ & $+6.91$ & $-4.25$ & $+0.16$ & $-2.83$ & $+7.53$ & $-2.87$ \\
\cline{2-16}      & \multicolumn{3}{c|}{$I_{P}$ ($\%$)} & $9.01$ & $7.11$ & $17.37$ & $7.96$ & $19.89$ & $6.46$ & $27.01$ & $10.26$ & $23.69$ & $0.05$ & $15.12$ & $0.91$ \\
    \hline
    \end{tabular}%
  \label{tab:RelativeToRandomLosslessDelayIdentical}%
\end{table*}%

\begin{table*}[t]
  \centering
  \caption{Relative Performance of Scheduling Policies over Non-Identical Lossless Links with Delay}
    \begin{tabular}{|r|c|c|c|c|c|c|c|c|c|c|c|}
    \hline
    \multicolumn{1}{|r|}{\multirow{9}[18]{*}{\vspace{1cm}\begin{sideways}Lossless\end{sideways}}} & \multicolumn{3}{c|}{Delay Model} & \multicolumn{4}{c|}{IV} & \multicolumn{4}{c|}{V} \\
\cline{2-12}    \multicolumn{1}{|r|}{} & \multicolumn{3}{c|}{Network Length} & \multicolumn{2}{c|}{$2$} & \multicolumn{2}{c|}{$8$} & \multicolumn{2}{c|}{$2$} & \multicolumn{2}{c|}{$8$} \\
\cline{2-12}    \multicolumn{1}{|r|}{} & \multicolumn{3}{c|}{Chunk Size} & $8$ & $32$ & $8$ & $32$ & $8$ & $32$ & $8$ & $32$ \\
\cline{2-12}    \multicolumn{1}{|r|}{} & \multirow{4}[8]{*}{\vspace{0.35cm}\begin{sideways}{\fontsize{4.5}{4.5}\selectfont Scheduling Policy}\end{sideways}} & RP & \multirow{4}[8]{*}{\vspace{0.45cm}$I_R$ ($\%$)} & $33.62$ & $6.40$ & $44.93$ & $8.80$ & $36.26$ & $7.49$ & $46.46$ & $9.14$ \\
\cline{3-3}\cline{5-12}    \multicolumn{1}{|r|}{} &   & LRF &   & $33.79$ & $8.92$ & $39.66$ & $11.10$ & $35.20$ & $9.08$ & $42.99$ & $11.69$ \\
\cline{3-3}\cline{5-12}    \multicolumn{1}{|r|}{} &   & MDF &   & $52.76$ & $13.12$ & $65.52$ & $25.48$ & $54.52$ & $15.99$ & $67.10$ & $26.12$ \\
\cline{3-3}\cline{5-12}    \multicolumn{1}{|r|}{} &   & MCMF &   & $46.71$ & $12.10$ & $56.28$ & $18.51$ & $41.74$ & $13.88$ & $59.75$ & $25.04$ \\
\cline{2-12}      & \multicolumn{3}{c|}{$I_{E}$ ($\%$)} & $-0.25$ & $-2.76$ & $+8.73$ & $-2.59$ & $+1.63$ & $-1.75$ & $+6.08$ & $-2.89$ \\
\cline{2-12}      & \multicolumn{3}{c|}{$I_{P}$ ($\%$)} & $11.35$ & $1.16$ & $21.14$ & $8.54$ & $21.92$ & $2.45$ & $18.24$ & $1.43$ \\
    \hline
    \end{tabular}%
  \label{tab:RelativeToRandomLosslessDelayNonIdentical}%
\end{table*}%

We first focus on the existing scheduling policies RP and LRF and their relative performance (the rows related to $I_R$ for RP and LRF, and $I_E$, in both Tables~\ref{tab:RelativeToRandomLosslessDelayIdentical} and~\ref{tab:RelativeToRandomLosslessDelayNonIdentical}). In the case of identical links (Table~\ref{tab:RelativeToRandomLosslessDelayIdentical}), for RP or LRF, as the mean and the variance of the delay become larger (i.e., moving from delay model \textmd{I}, with the smallest mean and variance, towards the delay model \textmd{III}, with the largest mean and variance), the parameter $I_R$ decreases, i.e., RP or LRF is more advantageous over the random scheduling policy for networks with delays with smaller mean and variance. For example, focusing on the results for RP, in the case with the chunk size $8$ and the network length $8$, for the delay model I, $I_R=48.69\%$, and for the delay models II and III, $I_R$ is reduced to $44.61\%$ and $43.25\%$, respectively. It is also worth noting that, for a given delay model, as the size of the chunks is decreased or the length of the network is increased, the parameter $I_R$ increases. More interestingly, the results of the second last row of the table ($I_E$) demonstrates that the relative performance of RP and LRF compared to each other also depends on the delay model. In particular, as the mean and the variance of the delay are increased, or the size of the chunks is decreased, or the length of the network is increased, the relative performance of RP and LRF changes to the benefit of RP. In particular, for a given delay model and network length, RP outperforms LRF for a sufficiently small chunk size.\footnote{It is worth noting that, in~\cite{XZWX:2008}, LRF and RP policies were compared over a number of network scenarios, and for the tested cases, it was concluded that LRF is superior to RP in terms of the expected delivery time. However, our simulation results on line networks demonstrate that the relative performance of these policies highly depends on the link model.} For example, considering the case for the chunk size $8$ and the network length $8$, and focusing on the comparison between LRF and RP over identical links (in Table~\ref{tab:RelativeToRandomLosslessDelayIdentical}), one can see that for the delay model \textmd{I}, RP is superior ($I_E=+6.54\%$). For the delay model \textmd{II}, the advantage of RP becomes more ($I_E=+6.91\%$), and for the delay model \textmd{III} with the largest mean and variance, RP is even more advantageous ($I_E=+7.53\%$). Similar trends can also be observed for the larger chunk size $32$. However, a closer look reveals that, for larger chunk sizes, the transition between the relative superiority of LRF over RP occurs at delays with larger mean and variance. For example, for the network length $8$ and the delay model \textmd{II}, RP is superior to LRF for the chunk size $8$ ($I_E=+6.92\%$), but for the larger chunk size $32$, LRF is still superior ($I_E=-4.25\%$). For delays with smaller mean and variance, LRF is superior to RP, since, in this case, there is a higher chance for a smaller difference between the set of packets at the receiving node at the time of transmission and that at the time of reception. Thus by giving the opportunity of transmission to a chunk with the smallest number of packets at the receiving node, there is a higher chance in balancing the number of packets for all the chunks. For delays with larger mean and variance, however, there is a higher chance for a bigger difference between the underlying sets, and hence, distributing the transmission opportunities over a larger set of chunks yields more balance.

In the case of non-identical links (Table~\ref{tab:RelativeToRandomLosslessDelayNonIdentical}), for a given delay model, similar to the case of identical links, the performance improvement of RP and LRF over random scheduling improves as the chuck size is reduced or the network length is increased. Also, as it can be seen for sufficiently small chunks and sufficiently large networks, RP outperforms LRF (i.e., $I_E$ is positive). For larger chunks or smaller networks, $I_E$ becomes smaller and for sufficiently large chunks and sufficiently small networks, $I_E$ crosses zero and becomes negative (i.e., LRF outperforms RP).

Similarly, by comparing MDF and MCMF, for fixed parameters $m$ and $\Delta$, and their relative performance compared to the random scheduling (the rows representing $I_R$ for MDF and MCMF, and $I_P$, in both Tables~\ref{tab:RelativeToRandomLosslessDelayIdentical} and~\ref{tab:RelativeToRandomLosslessDelayNonIdentical}), one can conclude that (i) for each scheduling policy, $I_R$ is decreased for delays with larger mean and variance, and for a given delay model, $I_R$ is increased for smaller chunks and larger networks; (ii) for (a given delay model with) delays with sufficiently small mean and variance, $I_P$ is increased (i.e., the performance gap between MDF and MCMF is increased) as the size of the chunks decreases or the length of the network increases. (Similarly, for sufficiently small chunks and sufficiently large networks, as the mean and the variance of the delay decrease, $I_P$ is decreased.) For example, for the delay model I, considering the chunk size $8$, for (smaller) network of length $2$, $I_P=9.01\%$, and for (larger) network of length $8$, $I_P$ is increased to $17.37\%$. Similarly, considering the network of length $2$, for the smaller chunk size of $8$, $I_P=9.01\%$, and for the larger chunk size of $32$, $I_P=7.11\%$. Similar comparison results hold true for the delay model II. One should however note that for the delay model III, with the largest mean and variance, similar trends do not seem to hold true. For example, considering the chunk size $8$, for the networks of length $2$, $I_P=23.69\%$, and it is reduced down to $15.12\%$ for the (larger) network of length $8$.

To justify the different trend for the delay model \textmd{III}, we note that the results of Tables~\ref{tab:RelativeToRandomLosslessDelayIdentical} and~\ref{tab:RelativeToRandomLosslessDelayNonIdentical} are based on fixed parameters $m$ and $\Delta$, and as a consequence, for delays with larger mean and variance, the approximation error in the calculation of the metrics is increased. In other words, for sufficiently large $m$ and $\Delta$ (and fixed chunk size and fixed network length), the (monotonically improving) trend of the relative performance of MDF compared to MCMF indeed does not change as the mean and the variance of delays are increased. To verify this claim, we have performed another experiment described below.

\begin{table*}[t]
  \centering
  \caption{Expected Delivery Time for MDF/MCMF with Sufficiently Large Parameters $m$ and $\Delta$}
    \begin{tabular}{|r|r|r|r|r|c|r|r|r|r|r|r|r|}
    \hline
    \multicolumn{1}{|r|}{\multirow{16}[28]{*}{\vspace{1.5cm}\begin{sideways}Lossless\end{sideways}}} & \multicolumn{4}{c|}{Network Length} & \multicolumn{8}{c|}{$2$}  \\
\cline{2-13}    \multicolumn{1}{|r|}{} & \multicolumn{4}{c|}{Chunk Size} & \multicolumn{8}{c|}{$4$} \\
\cline{2-13}    \multicolumn{1}{|r|}{} & \multicolumn{4}{c|}{Delay Model} & \multicolumn{4}{c|}{II} & \multicolumn{4}{c|}{III} \\
\cline{2-13}    \multicolumn{1}{|r|}{} & \multicolumn{1}{r|}{\multirow{13}[22]{*}{\vspace{0.85cm}\begin{sideways}Scheduling Policy\end{sideways}}} & \multicolumn{1}{c|}{\multirow{6}[10]{*}{\vspace{0.275cm}MDF}} & \multicolumn{2}{c|}{\multirow{2}[2]{*}{\backslashbox{$\Delta$}{$m$}}} & \multirow{2}[2]{*}{$2$} & \multicolumn{1}{c|}{\multirow{2}[2]{*}{$3$}} & \multicolumn{1}{c|}{\multirow{2}[2]{*}{$4$}} & \multicolumn{1}{c|}{\multirow{2}[2]{*}{$5$}} & \multicolumn{1}{c|}{\multirow{2}[2]{*}{$2$}} & \multicolumn{1}{c|}{\multirow{2}[2]{*}{$3$}} & \multicolumn{1}{c|}{\multirow{2}[2]{*}{$4$}} & \multicolumn{1}{c|}{\multirow{2}[2]{*}{$5$}} \\
    \multicolumn{1}{|r|}{} & \multicolumn{1}{r|}{} & \multicolumn{1}{c|}{} & \multicolumn{2}{c|}{} &   & \multicolumn{1}{c|}{} & \multicolumn{1}{c|}{} & \multicolumn{1}{c|}{} & \multicolumn{1}{c|}{} & \multicolumn{1}{c|}{} & \multicolumn{1}{c|}{} & \multicolumn{1}{c|}{} \\
\cline{4-13}    \multicolumn{1}{|r|}{} & \multicolumn{1}{r|}{} & \multicolumn{1}{c|}{} & \multicolumn{2}{c|}{$2$} & $15.44$ & \multicolumn{1}{c|}{$15.34$} & \multicolumn{1}{c|}{$15.31$} & \multicolumn{1}{c|}{$15.30$} & \multicolumn{1}{c|}{$19.74$} & \multicolumn{1}{c|}{$19.64$} & \multicolumn{1}{c|}{$19.50$} & \multicolumn{1}{c|}{$19.35$} \\
\cline{4-13}    \multicolumn{1}{|r|}{} & \multicolumn{1}{r|}{} & \multicolumn{1}{c|}{} & \multicolumn{2}{c|}{$3$} & $15.33$ & \multicolumn{1}{c|}{$15.10$} & \multicolumn{1}{c|}{$14.98$} & \multicolumn{1}{c|}{$14.97$} & \multicolumn{1}{c|}{$19.54$} & \multicolumn{1}{c|}{$18.80$} & \multicolumn{1}{c|}{$18.75$} & \multicolumn{1}{c|}{$18.68$} \\
\cline{4-13}    \multicolumn{1}{|r|}{} & \multicolumn{1}{r|}{} & \multicolumn{1}{c|}{} & \multicolumn{2}{c|}{$4$} & $15.07$ & \multicolumn{1}{c|}{$14.98$} & \multicolumn{1}{c|}{$14.97$} & \multicolumn{1}{c|}{$14.97$} & \multicolumn{1}{c|}{$19.37$} & \multicolumn{1}{c|}{$18.73$} & \multicolumn{1}{c|}{$18.69$} & \multicolumn{1}{c|}{$18.65$} \\
\cline{4-13}    \multicolumn{1}{|r|}{} & \multicolumn{1}{r|}{} & \multicolumn{1}{c|}{} & \multicolumn{2}{c|}{$5$} & $14.99$ & \multicolumn{1}{c|}{$14.97$} & \multicolumn{1}{c|}{$14.97$} & \multicolumn{1}{c|}{$14.97$} & \multicolumn{1}{c|}{$19.20$} & \multicolumn{1}{c|}{$18.68$} & \multicolumn{1}{c|}{$18.65$} & \multicolumn{1}{c|}{$18.65$} \\
\cline{3-13}    \multicolumn{1}{|r|}{} & \multicolumn{1}{r|}{} & \multicolumn{1}{c|}{\multirow{6}[10]{*}{\vspace{0.275cm}MCMF}} & \multicolumn{2}{c|}{\multirow{2}[2]{*}{\backslashbox{$\Delta$}{$m$}}} & \multirow{2}[2]{*}{$2$} & \multicolumn{1}{c|}{\multirow{2}[2]{*}{$3$}} & \multicolumn{1}{c|}{\multirow{2}[2]{*}{$4$}} & \multicolumn{1}{c|}{\multirow{2}[2]{*}{$5$}} & \multicolumn{1}{c|}{\multirow{2}[2]{*}{$3$}} & \multicolumn{1}{c|}{\multirow{2}[2]{*}{$3$}} & \multicolumn{1}{c|}{\multirow{2}[2]{*}{$4$}} & \multicolumn{1}{c|}{\multirow{2}[2]{*}{$5$}} \\
    \multicolumn{1}{|r|}{} & \multicolumn{1}{r|}{} & \multicolumn{1}{c|}{} & \multicolumn{2}{c|}{} &   & \multicolumn{1}{c|}{} & \multicolumn{1}{c|}{} & \multicolumn{1}{c|}{} & \multicolumn{1}{c|}{} & \multicolumn{1}{c|}{} & \multicolumn{1}{c|}{} & \multicolumn{1}{c|}{} \\
\cline{4-13}    \multicolumn{1}{|r|}{} & \multicolumn{1}{r|}{} & \multicolumn{1}{c|}{} & \multicolumn{2}{c|}{$2$} & $15.89$ & \multicolumn{1}{c|}{$15.73$} & \multicolumn{1}{c|}{$15.54$} & \multicolumn{1}{c|}{$15.44$} & \multicolumn{1}{c|}{$20.15$} & \multicolumn{1}{c|}{$19.88$} & \multicolumn{1}{c|}{$19.64$} & \multicolumn{1}{c|}{$19.51$} \\
\cline{4-13}    \multicolumn{1}{|r|}{} & \multicolumn{1}{r|}{} & \multicolumn{1}{c|}{} & \multicolumn{2}{c|}{$3$} & $15.66$ & \multicolumn{1}{c|}{$15.37$} & \multicolumn{1}{c|}{$15.34$} & \multicolumn{1}{c|}{$15.33$} & \multicolumn{1}{c|}{$19.84$} & \multicolumn{1}{c|}{$19.24$} & \multicolumn{1}{c|}{$19.14$} & \multicolumn{1}{c|}{$19.13$} \\
\cline{4-13}    \multicolumn{1}{|r|}{} & \multicolumn{1}{r|}{} & \multicolumn{1}{c|}{} & \multicolumn{2}{c|}{$4$} & $15.47$ & \multicolumn{1}{c|}{$15.36$} & \multicolumn{1}{c|}{$15.33$} & \multicolumn{1}{c|}{$15.33$} & \multicolumn{1}{c|}{$19.63$} & \multicolumn{1}{c|}{$19.24$} & \multicolumn{1}{c|}{$19.14$} & \multicolumn{1}{c|}{$19.12$} \\
\cline{4-13}    \multicolumn{1}{|r|}{} & \multicolumn{1}{r|}{} & \multicolumn{1}{c|}{} & \multicolumn{2}{c|}{$5$} & $15.38$ & \multicolumn{1}{c|}{$15.35$} & \multicolumn{1}{c|}{$15.33$} & \multicolumn{1}{c|}{$15.33$} & \multicolumn{1}{c|}{$19.48$} & \multicolumn{1}{c|}{$19.24$} & \multicolumn{1}{c|}{$19.13$} & \multicolumn{1}{c|}{$19.12$} \\
\cline{3-13}      &   & \multicolumn{3}{c|}{Random} & \multicolumn{4}{c|}{$21.88$} & \multicolumn{4}{c|}{$24.47$} \\
    \hline
    \end{tabular}%
  \label{tab:UnboundedPars}%
\end{table*}%

Consider the transmission of a message of size $8$ over a line network of length $2$ with CC where the chunk size is $4$ (two chunks). In this experiment, we only consider the delay models II and III. For both MDF and MCMF policies, the parameters $m$ and $\Delta$ vary between $2$ and $5$, i.e., $2\leq m\leq 5$, and $2\leq \Delta\leq 5$. For each delay model, $100$ network realizations are simulated, and for each pair of choices of $m$ and $\Delta$, CC with MDF or MCMF scheduling policy is applied to each network realization for $100$ trials. The expected delivery time, for each case, is the average of the delivery time over all the simulated realizations of the code and the network realization. These results are presented in Table~\ref{tab:UnboundedPars}. The corresponding relative performances, $I_R$ (for each policy) and $I_P$, are also listed in Table~\ref{tab:RelativePerformanceUnboundedPars}. The results in Table~\ref{tab:UnboundedPars} demonstrate that, for MDF or MCMF, the expected delivery time reaches a limit as $m$ and $\Delta$ grow large (i.e., the approximation error in the metrics can be made sufficiently small by choosing $m$ and $\Delta$ sufficiently large). In the case of delay model II, for MDF and MCMF, this limit is equal to $14.97$ and $15.33$, respectively; and in the case of delay model III, it is equal to $18.65$ and $19.12$, respectively. It can be seen that the limit of the expected delivery time itself does change with the delay model. For each scheduling policy, MDF or MCMF, with sufficiently large $m$ and $\Delta$, as can be seen in Table~\ref{tab:RelativePerformanceUnboundedPars}, $I_R$ decreases (and it is not constant) as the mean and the variance of the delay increases. Furthermore, MDF becomes more advantageous compared to MCMF for delays with larger mean and variance (i.e., $I_P$ becomes larger as the mean and the variance of the delay are increased).

\begin{table}[t]
  \centering
  \caption{Relative Performance of MDF/MCMF with Sufficiently Large Parameters $m$ and $\Delta$}
    \begin{tabular}{|c|c|c|c|c|c|}
    \hline
    \multirow{6}[12]{*}{\vspace{0.5cm}\begin{sideways}Lossless\end{sideways}} & \multicolumn{3}{c|}{Network Length} & \multicolumn{2}{c|}{$2$} \\
\cline{2-6}      & \multicolumn{3}{c|}{Chunk Size} & \multicolumn{2}{c|}{$4$} \\
\cline{2-6}      & \multicolumn{3}{c|}{Delay Model} & II & III \\
\cline{2-6}      & \multirow{2}[4]{*}{\begin{sideways}{\fontsize{4.5}{4.5}\selectfont Policy}\end{sideways}} & MDF & \multirow{2}[4]{*}{\vspace{0.15cm}$I_R$ ($\%$)} & $31.58$ & $23.78$ \\
\cline{3-3}\cline{5-6}      &   & MCMF &   & $29.93$ & $21.86$ \\
\cline{2-6}      & \multicolumn{3}{c|}{$I_P$ ($\%$)} & $2.34$ & $2.45$ \\
    \hline
    \end{tabular}%
  \label{tab:RelativePerformanceUnboundedPars}%
\end{table}%

\begin{table}[t!]
  \centering
  \caption{Parameters of Loss Models Used in the Simulations}
    \begin{tabular}{|c|c|c|c|c|}
    \hline
    \multicolumn{2}{|c|}{Network Length} & \multicolumn{3}{c|}{$L$} \\
    \hline
    \multicolumn{2}{|c|}{Loss Model} & I & II & III \\
    \hline\hline
    \multicolumn{2}{|c|}{$p_e^{(i)}$} & $\frac{1}{3}$ & $\frac{1}{3}\cdot\frac{i}{L}$ & $\frac{1}{3}\cdot\frac{L-i+1}{L}$ \\
    \hline
    \end{tabular}%
  \label{tab:LossModelsPars}%
\end{table}%

\subsection{Lossy Links with Unit Delays}\label{subsec:LossyLinksWithUnitDelay}
The scenarios considered in this part are very similar to those in Section~\ref{subsec:LosslessLinksWithDelay}, except that the links are lossy and their loss model is specified in Table~\ref{tab:LossModelsPars}, and that the delay model of each link is the unit-delay model. (In particular, the loss model I considers a case with identical links with erasure probability $\frac{1}{3}$, and the two loss models II and III represent two cases with non-identical links.)

\begin{table}[t]
  \centering
  \caption{Expected Delivery Time for Various Scheduling Policies over Identical Lossy Links with Unit Delays}
    \begin{tabular}{|c|r|c|c|c|c|c|}
    \hline
    \multirow{10}[20]{*}{\vspace{1cm}\begin{sideways}Unit-Delay\end{sideways}} & \multicolumn{2}{c|}{Loss Model} & \multicolumn{4}{c|}{I} \\
\cline{2-7}      & \multicolumn{2}{c|}{Network Length} & \multicolumn{2}{c|}{$2$} & \multicolumn{2}{c|}{$8$} \\
\cline{2-7}      & \multicolumn{2}{c|}{Chunk Size} & $8$ & $32$ & $8$ & $32$ \\
\cline{2-7}      & \multicolumn{1}{r|}{\multirow{5}[10]{*}{\vspace{0.5cm}\begin{sideways}{\fontsize{4.5}{4.5}\selectfont Scheduling Policy}\end{sideways}}} & Random & $321.21$ & $181.03$ & $656.89$ & $312.18$ \\
\cline{3-7}      &   & RP & $191.00$ & $163.31$ & $322.59$ & $269.00$ \\
\cline{3-7}      &   & LRF & $192.14$ & $162.25$ & $334.02$ & $265.20$ \\
\cline{3-7}      &   & MDF & $148.47$ & $148.47$ & $224.46$ & $224.46$ \\
\cline{3-7}      &   & MCMF & $148.47$ & $148.47$ & $224.46$ & $224.46$ \\
\cline{2-7}      & \multicolumn{2}{c|}{$I_1$ ($\%$)} & $22.26$ & $8.49$ & $30.41$ & $15.36$ \\
\cline{2-7}      & \multicolumn{2}{c|}{$I_2$ ($\%$)} & $22.26$ & $8.49$ & $30.41$ & $15.36$ \\
    \hline
    \end{tabular}%
  \label{tab:ExpSuccessTimeLossyUnitDelayIdentical}%
\end{table}%

\begin{table*}[t]
  \centering
  \caption{Expected Delivery Time for Various Scheduling Policies over Non-Identical Lossy Links with Unit Delays}
    \begin{tabular}{|c|r|c|c|c|c|c|c|c|c|c|}
    \hline
    \multirow{10}[20]{*}{\vspace{1cm}\begin{sideways}Unit-Delay\end{sideways}} & \multicolumn{2}{c|}{Loss Model} & \multicolumn{4}{c|}{II} & \multicolumn{4}{c|}{III} \\
\cline{2-11}      & \multicolumn{2}{c|}{Network Length} & \multicolumn{2}{c|}{$2$} & \multicolumn{2}{c|}{$8$} & \multicolumn{2}{c|}{$2$} & \multicolumn{2}{c|}{$8$} \\
\cline{2-11}      & \multicolumn{2}{c|}{Chunk Size} & $8$ & $32$ & $8$ & $32$ & $8$ & $32$ & $8$ & $32$ \\
\cline{2-11}      & \multicolumn{1}{r|}{\multirow{5}[10]{*}{\vspace{0.5cm}\begin{sideways}{\fontsize{4.5}{4.5}\selectfont Scheduling Policy}\end{sideways}}} & Random & $269.87$ & 159.19 & 478.79 & 225.47 & 280.45 & 157.52 & 494.84 & 224.17 \\
\cline{3-11}      &   & RP & $169.75$ & $144.27$ & $239.15$ & $195.80$ & $161.50$ & $141.35$ & $227.77$ & $190.32$ \\
\cline{3-11}      &   & LRF & $171.92$ & $142.74$ & $248.04$ & $193.79$ & $163.55$ & $140.72$ & $232.04$ & $191.44$ \\
\cline{3-11}      &   & MDF & $131.91$ & $131.91$ & $158.85$ & $158.85$ & $131.40$ & $131.40$ & $157.26$ & $157.26$ \\
\cline{3-11}      &   & MCMF & $131.91$ & $131.91$ & $158.85$ & $158.85$ & $131.40$ & $131.40$ & $157.26$ & $157.26$ \\
\cline{2-11}      & \multicolumn{2}{c|}{$I_1$ ($\%$)} & $22.29$ & $7.58$ & $33.57$ & $18.02$ & $18.63$ & $6.62$ & $30.95$ & $17.37$ \\
\cline{2-11}      & \multicolumn{2}{c|}{$I_2$ ($\%$)} & $22.29$ & $7.58$ & $33.57$ & $18.02$ & $18.63$ & $6.62$ & $30.95$ & $17.37$ \\
    \hline
    \end{tabular}%
  \label{tab:ExpSuccessTimeLossyUnitDelayNonIdentical}%
\end{table*}%

\begin{table}[t]
  \centering
  \caption{Relative Performance of Scheduling Policies over Identical Lossy Links with Unit Delays}
    \begin{tabular}{|r|c|c|c|c|c|c|c|}
    \hline
    \multicolumn{1}{|r|}{\multirow{9}[18]{*}{\vspace{1cm}\begin{sideways}Unit-Delay\end{sideways}}} & \multicolumn{3}{c|}{Loss Model} & \multicolumn{4}{c|}{I} \\
\cline{2-8}    \multicolumn{1}{|r|}{} & \multicolumn{3}{c|}{Network Length} & \multicolumn{2}{c|}{$2$} & \multicolumn{2}{c|}{$8$} \\
\cline{2-8}    \multicolumn{1}{|r|}{} & \multicolumn{3}{c|}{Chunk Size} & $8$ & $32$ & $8$ & $32$ \\
\cline{2-8}    \multicolumn{1}{|r|}{} & \multirow{4}[8]{*}{\vspace{0.35cm}\begin{sideways}{\fontsize{4.5}{4.5}\selectfont Scheduling Policy}\end{sideways}} & RP & \multirow{4}[8]{*}{\vspace{0.45cm}$I_R$ ($\%$)} & $40.53$ & $9.78$ & $50.89$ & $13.83$ \\
\cline{3-3}\cline{5-8}    \multicolumn{1}{|r|}{} &   & LRF &   & $40.18$ & $10.37$ & $49.15$ & $15.04$ \\
\cline{3-3}\cline{5-8}    \multicolumn{1}{|r|}{} &   & MDF &   & $53.77$ & $17.98$ & $65.82$ & $28.09$ \\
\cline{3-3}\cline{5-8}    \multicolumn{1}{|r|}{} &   & MCMF &   & $53.77$ & $17.98$ & $65.82$ & $28.09$ \\
\cline{2-8}      & \multicolumn{3}{c|}{$I_{E}$ ($\%$)} & $+0.59$ & $-0.65$ & $+3.42$ & $-1.43$ \\
\cline{2-8}      & \multicolumn{3}{c|}{$I_{P}$ ($\%$)} & $0.00$ & $0.00$ & $0.00$ & $0.00$ \\
    \hline
    \end{tabular}%
  \label{tab:RelativeToRandomLossyUnitDelayIdentical}%
\end{table}%

\begin{table*}[t]
  \centering
  \caption{Relative Performance of Scheduling Policies over Non-Identical Lossy Links with Unit Delays}
    \begin{tabular}{|r|c|c|c|c|c|c|c|c|c|c|c|}
    \hline
    \multicolumn{1}{|r|}{\multirow{9}[18]{*}{\vspace{1cm}\begin{sideways}Unit-Delay\end{sideways}}} & \multicolumn{3}{c|}{Loss Model} & \multicolumn{4}{c|}{II} & \multicolumn{4}{c|}{III} \\
\cline{2-12}    \multicolumn{1}{|r|}{} & \multicolumn{3}{c|}{Network Length} & \multicolumn{2}{c|}{$2$} & \multicolumn{2}{c|}{$8$} & \multicolumn{2}{c|}{$2$} & \multicolumn{2}{c|}{$8$} \\
\cline{2-12}    \multicolumn{1}{|r|}{} & \multicolumn{3}{c|}{Chunk Size} & $8$ & $32$ & $8$ & $32$ & $8$ & $32$ & $8$ & $32$ \\
\cline{2-12}    \multicolumn{1}{|r|}{} & \multirow{4}[8]{*}{\vspace{0.35cm}\begin{sideways}{\fontsize{4.5}{4.5}\selectfont Scheduling Policy}\end{sideways}} & RP & \multirow{4}[8]{*}{\vspace{0.45cm}$I_R$ ($\%$)} & $37.09$ & $9.37$ & $50.05$ & $13.15$ & $42.41$ & $10.26$ & $53.97$ & $15.10$ \\
\cline{3-3}\cline{5-12}    \multicolumn{1}{|r|}{} &   & LRF &   & $36.29$ & $10.33$ & $48.19$ & $14.05$ & $41.68$ & $10.66$ & $53.10$ & $14.60$ \\
\cline{3-3}\cline{5-12}    \multicolumn{1}{|r|}{} &   & MDF &   & $51.12$ & $17.13$ & $66.82$ & $29.54$ & $53.14$ & $16.58$ & $68.22$ & $29.84$ \\
\cline{3-3}\cline{5-12}    \multicolumn{1}{|r|}{} &   & MCMF &   & $51.12$ & $17.13$ & $66.82$ & $29.54$ & $53.14$ & $16.58$ & $68.22$ & $29.84$ \\
\cline{2-12}      & \multicolumn{3}{c|}{$I_{E}$ ($\%$)} & $+1.26$ & $-1.07$ & $+3.58$ & $-1.03$ & $+1.25$ & $-0.44$ & $+1.84$ & $+0.58$ \\
\cline{2-12}      & \multicolumn{3}{c|}{$I_{P}$ ($\%$)} & $0.00$ & $0.00$ & $0.00$ & $0.00$ & $0.00$ & $0.00$ & $0.00$ & $0.00$ \\
    \hline
    \end{tabular}%
  \label{tab:RelativeToRandomLossyUnitDelayNonIdentical}%
\end{table*}%

Tables~\ref{tab:ExpSuccessTimeLossyUnitDelayIdentical} and~\ref{tab:ExpSuccessTimeLossyUnitDelayNonIdentical} list the expected delivery time for each scheduling policy in the case of identical and non-identical lossy links with unit delay, respectively. Based on the results in these tables, the relative performance of LRF and RP (or MDF and MCMF) compared to each other and compared to the random scheduling policy, are also listed in Tables~\ref{tab:RelativeToRandomLossyUnitDelayIdentical} and~\ref{tab:RelativeToRandomLossyUnitDelayNonIdentical}.

Based on the results in the tables, we observe the followings: (i) MDF and MCMF are always superior to RP and LRF ($I_1=I_2>0$); (ii) The advantage of MDF/MCMF over RP/LRF is larger for smaller chunks and larger networks (i.e., $I_1$ ($=I_2$) increases, as the size of the chunks is decreased, or as the length of the network is increased); (iii) The relative performance of RP vs. LRF depends on the chunk size and network length. For sufficiently small chunk sizes and sufficiently large networks, the relative performance of RP with respect to LRF improves (i.e., $I_E$ is increased) as the chunk size is reduced or as the network length is increased. This trend however reverses for sufficiently large chunks or sufficiently small networks; (iv) In all lossy scenarios with unit delay, MDF and MCMF perform the same ($I_P=0$). This is expected based on Lemma~\ref{lem:1}.

\subsection{On the (Near) Optimality of the MDF Scheduling Policy over Line Networks}\label{subsec:OptimalityMDF}
To verify the fact that the MDF scheduling policy is (near) optimal in the sense of minimizing the expected delivery time over line networks, we have performed the following experiment.

Consider a simple line network of length $1$ (one transmitting node and one receiving node). The link is assumed to be lossless, and two cases with two different delay models I and II are considered. The message size is $8$, and we consider chunks of size $4$ (i.e., the case of CC with two chunks). The parameters $m$ and $\Delta$ vary between $2$ and $4$, i.e., $2\leq m\leq 4$, and $2\leq \Delta\leq 4$. For each delay model, $100$ network realizations are simulated, and for each choice of parameters $m$ and $\Delta$, CC with the MDF scheduling policy is applied to each network realization. For each case (i.e., a given network realization), we consider all the possible choices of chunks (to be selected) at each transmission time, and set the maximum transmission time equal to $N_{\text{max}}$, for some $N_{\text{max}}\in\{16,32\}$ (i.e., we consider all the possible sequences of the chunk indices till the time $N_{\text{max}}$). We further focus on those sequences for which the (decoding) success occurs (i.e., for each chunk, among all the packets transmitted till the time $N_{\text{max}}$ by the transmitting node, $4$ innovative packets pertaining to that chunk are received at the receiving node). For each such successful sequence, we record the choice (index) of the chunk selected (to be transmitted) at the time $N_0$, for some $N_0\in\{4,8\}$ (we only pick two values of $N_0$ as considering all the possible values between $1$ and $N_{\text{max}}$ is too complex). For each chunk $\omega\in\{1,2\}$, we calculate the average of delivery times over all those (successful) sequences in which the chunk $\omega$ is selected at the time $N_0$, and denote it by $E(\omega)$. We also calculate (approximate) the two distances $d_{N_0}(1)$ and $d_{N_0}(2)$ between the two metrics vectors $\boldsymbol{x}_{N_0}(1)$ and $\boldsymbol{x}_{N_0}(2)$ (defined in Section~\ref{subsec:MDF}), and the target vector $[4,4]$, respectively. Let $\omega_{E}\doteq\arg\min_{\omega} E(\omega)$ and $\omega_{d}\doteq\arg\min_{\omega}d^{(1)}_{N_0}(\omega)$. Let us define an indicator variable $I$ such that $I=1$, if $\omega_E=\omega_d$; and $I=0$, otherwise. Note that $\omega_d$ is the chunk which will be selected based on the MDF scheduling policy at the transmission time $N_0$, and $\omega_E$ is the chunk whose selection at the transmission time minimizes the expected delivery time. Therefore, if $I=1$, then both events coincide, and in other words, the MDF policy minimizes the expected delivery time. Table~\ref{tab:OptimalityMDF} lists the average of the indicator variables $I$ (in percentage) for all the $100$ network realizations, for each delay model and each pair of parameters $m$ and $\Delta$. The closer is the expected value of $I$ to $100\%$, the closer the MDF policy would be to minimize the expected delivery time. As can be seen in the table, for each delay model, for sufficiently large values of $m$ and $\Delta$ (and for sufficiently large choice of $N_{\text{max}}$), $I$ is equal to $100\%$. This is particularly the case for sufficiently large $N_{\text{max}}$, and for sufficiently small $N_0$ in comparison with $N_{\text{max}}$. This would indicate that the MDF policy with sufficiently large $m$ and $\Delta$ (depending on the delay model parameters, the chunk size and the network length) achieves the minimum expected delivery time in each case.

\begin{table*}[t]
  \centering
  \caption{On the Optimality of the MDF Scheduling Policy}
    \begin{tabular}{|c|c|c|c|c|c|c|c|c|c|c|c|c|c|c|c|c|}
    \hline
    \multirow{14}[24]{*}{\vspace{1.5cm}\begin{sideways}Lossless\end{sideways}} & \multicolumn{4}{c|}{Network Length} & \multicolumn{12}{c|}{$1$} \\
\cline{2-17}      & \multicolumn{4}{c|}{Chunk Size} & \multicolumn{12}{c|}{$4$}  \\
\cline{2-17}      & \multicolumn{4}{c|}{$N_0$} & \multicolumn{6}{c|}{$4$} & \multicolumn{6}{c|}{$8$} \\
\cline{2-17}      & \multicolumn{4}{c|}{$N_{\text{max}}$} & \multicolumn{3}{c|}{$16$} & \multicolumn{3}{c|}{$32$} & \multicolumn{3}{c|}{$16$} & \multicolumn{3}{c|}{$32$} \\
\cline{2-17}      & \multirow{10}[16]{*}{\vspace{1cm}\begin{sideways}Delay Model\end{sideways}} & \multirow{5}[8]{*}{\vspace{0.5cm} I} & \multicolumn{2}{c|}{\multirow{2}[2]{*}{\backslashbox{$\Delta$}{$m$}}} & \multirow{2}[2]{*}{$2$} & \multirow{2}[2]{*}{$3$} & \multirow{2}[2]{*}{$4$} & \multirow{2}[2]{*}{$2$} & \multirow{2}[2]{*}{$3$} & \multirow{2}[2]{*}{$4$} & \multirow{2}[2]{*}{$2$} & \multirow{2}[2]{*}{$3$} & \multirow{2}[2]{*}{$4$} & \multirow{2}[2]{*}{$2$} & \multirow{2}[2]{*}{$3$} & \multirow{2}[2]{*}{$4$} \\
      &   &   & \multicolumn{2}{c|}{} &   &   &   &   &   &   &   &   &   &   &   & \\
\cline{4-17}      &   &   & \multicolumn{2}{c|}{$2$} & $98.5$ & $100$ & $100$ & $100$ & $100$ & $100$ & $97.5$ & $100$ & $100$ & $100$ & $100$ & $100$ \\
\cline{4-17}      &   &   & \multicolumn{2}{c|}{$3$} & $100$ & $100$ & $100$ & $100$ & $100$ & $100$ & $100$ & $100$ & $100$ & $100$ & $100$ & $100$ \\
\cline{4-17}      &   &   & \multicolumn{2}{c|}{$4$} & $100$ & $100$ & $100$ & $100$ & $100$ & $100$ & $100$ & $100$ & $100$ & $100$ & $100$ & $100$ \\
\cline{3-17}      &   & \multirow{5}[8]{*}{\vspace{0.5cm} II} & \multicolumn{2}{c|}{\multirow{2}[2]{*}{\backslashbox{$\Delta$}{$m$}}} & \multirow{2}[2]{*}{$2$} & \multirow{2}[2]{*}{$3$} & \multirow{2}[2]{*}{$4$} & \multirow{2}[2]{*}{$2$} & \multirow{2}[2]{*}{$3$} & \multirow{2}[2]{*}{$4$} & \multirow{2}[2]{*}{$2$} & \multirow{2}[2]{*}{$3$} & \multirow{2}[2]{*}{$4$} & \multirow{2}[2]{*}{$2$} & \multirow{2}[2]{*}{$3$} & \multirow{2}[2]{*}{$4$} \\
      &   &   & \multicolumn{2}{c|}{} &   &   &   &   &   &   &   &   &   &   &   &  \\
\cline{4-17}      &   &   & \multicolumn{2}{c|}{$2$} & $90$ & $90$ & $90$ & $100$ & $100$ & $100$ & $80.9$ & $87.0$ & $93.2$ & $84.6$ & $94.0$ & $99.6$ \\
\cline{4-17}      &   &   & \multicolumn{2}{c|}{$3$} & $90$ & $90$ & $90$ & $100$ & $100$ & $100$ & $80.9$ & $93.4$ & $93.7$ & $89.2$ & $99.5$ & $100$ \\
\cline{4-17}      &   &   & \multicolumn{2}{c|}{$4$} & $90$ & $90$ & $90$ & $100$ & $100$ & $100$ & $78.7$ & $93.2$ & $93.7$ & $90.9$ & $99.5$ & $100$ \\
    \hline
    \end{tabular}%
  \label{tab:OptimalityMDF}%
\end{table*}%


\section{Conclusion}
We proposed two feedback-based policies, called minimum-distance-first (MDF) and minimum current metric first (MCMF), for scheduling the chunks in chunked codes over networks with delay and loss, where MCMF is a low-complexity version of MDF with a rather small loss in the performance. In contrast with the existing scheduling policies, random push (RP) and local-rarest-first (LRF), that prioritize the chunks based on the number of innovative received packets over the links (by using the feedback information) up to the time of the new transmission, MDF and MCMF incorporate the loss and delay models in the scheduling process, and operate based on the expected number of innovative successful packet transmissions at each node of the network prior to the next and current transmission time, respectively. To study the performance of the proposed scheduling policies in comparison with the existing ones, we used the log-normal and the unit delay models as well as the lossless and the Bernoulli loss model, over line networks. Our simulations showed that MDF and MCMF significantly outperform (by up to about $46\%$ and $34.72\%$, respectively, for the tested cases) the existing policies of RP and LRF in terms of the expected time required for all the chunks to be decodable. The performance improvements are specially larger for smaller chunks and larger networks. Such scenarios are of particular practical interest as smaller chunks translate to lower coding costs, and larger networks are just a fact of life with the continuous increase in the number of communication devices. The improvements come at the cost of more computations, and a slight increase in the amount of feedback. Our results also indicate that the relative performance of RP vs. LRF changes depending on the delay and loss models, the length of network and the size of chunks.

The performance comparison of the proposed and the existing scheduling policies over more general network topologies is also an interesting problem that requires more investigation. While such an investigation was not performed in this work, we expect that the proposed policies still outperform RP and LRF policies over more general network topologies. This would be particularly the case, if the network topology allows for the inclusion of the transmission times of the delayed packets of all the adjacent transmitter neighbors of a receiving node in the decision making process at each such transmitting neighbor.

\bibliographystyle{IEEEtran}
\bibliography{IEEEabrv,Refs}

\begin{thebibliography}{10}
\providecommand{\url}[1]{#1}
\csname url@samestyle\endcsname
\providecommand{\newblock}{\relax}
\providecommand{\bibinfo}[2]{#2}
\providecommand{\BIBentrySTDinterwordspacing}{\spaceskip=0pt\relax}
\providecommand{\BIBentryALTinterwordstretchfactor}{4}
\providecommand{\BIBentryALTinterwordspacing}{\spaceskip=\fontdimen2\font plus
\BIBentryALTinterwordstretchfactor\fontdimen3\font minus
  \fontdimen4\font\relax}
\providecommand{\BIBforeignlanguage}[2]{{%
\expandafter\ifx\csname l@#1\endcsname\relax
\typeout{** WARNING: IEEEtran.bst: No hyphenation pattern has been}%
\typeout{** loaded for the language `#1'. Using the pattern for}%
\typeout{** the default language instead.}%
\else
\language=\csname l@#1\endcsname
\fi
#2}}
\providecommand{\BIBdecl}{\relax}
\BIBdecl

\bibitem{GR:2005}
C.~Gkantsidis and P.~Rodriguez, ``Network {C}oding for {L}arge {S}cale
  {C}ontent {D}istribution,'' in \emph{Proc. IEEE INFOCOM'05}, vol.~4, March
  2005, pp. 2235--2245.

\bibitem{GMR:2006}
C.~Gkantsidis, J.~Miller, and P.~Rodriguez, ``Comprehensive {V}iew of a {L}ive
  {N}etwork {C}oding {P2P} {S}ystem,'' in \emph{Proc. Internet Measurement
  Conference, IMC'06}, 2006.

\bibitem{GMR2:2006}
------, ``Anatomy of a {P2P} {C}ontent {D}istribution {S}ystem with {N}etwork
  {C}oding,'' in \emph{Proc. Int. Workshop on {P2P} Systems, IPTPS'06}, 2006.

\bibitem{NL:2007}
D.~Niu and B.~Li, ``On the {R}esilience-{C}omplexity {T}radeoff of {N}etwork
  {C}oding in {D}ynamic {P2P} {N}etworks,'' in \emph{Proc. 15th IEEE Int.
  Workshop on Quality of Service, IWQoS'07}, June 2007, pp. 38--46.

\bibitem{MHL:2006}
P.~Maymounkov, N.~Harvey, and D.~Lun, ``Methods for {E}fficient {N}etwork
  {C}oding,'' in \emph{Proc. 44th Annual Allerton Conference on Communication
  Control and Computing}, 2006, pp. 482--491.

\bibitem{SZK:2009}
D.~Silva, W.~Zeng, and F.~Kschischang, ``Sparse {N}etwork {C}oding with
  {O}verlapping {C}lasses,'' in \emph{Proc. Workshop on Network Coding, Theory
  and Applications, NetCod'09}, June 2009, pp. 74--79.

\bibitem{HB:2010}
A.~Heidarzadeh and A.~Banihashemi, ``Overlapped {C}hunked {N}etwork {C}oding,''
  in \emph{Proc. IEEE Info. Theory Workshop, ITW'10}, Jan. 2010.

\bibitem{HB:2011}
------, ``Analysis of {O}verlapped {C}hunked {C}odes with {S}mall {C}hunks over
  {L}ine {N}etworks,'' in \emph{Proc. IEEE Int. Symp. on Info. Theory,
  ISIT'11}, Aug. 2011, pp. 864--868.

\bibitem{WL:2007}
M.~Wang and B.~Li, ``{R2}: {R}andom {P}ush with {R}andom {N}etwork {C}oding in
  {L}ive {P}eer-to-{P}eer {S}treaming,'' \emph{IEEE Journal on Selected Areas
  in Communications}, vol.~25, no.~9, pp. 1655--1666, Dec. 2007.

\bibitem{XZWX:2008}
J.~Xu, J.~Zhao, X.~Wang, and X.~Xue, ``Swifter: {C}hunked {N}etwork {C}oding
  for {P}eer-to-{P}eer {C}ontent {D}istribution,'' in \emph{Proc. IEEE Int.
  Conference on Communications, ICC'08}, May 2008, pp. 5603--5608.

\bibitem{Paxson:1997}
V.~E. Paxson, ``Measurements and {A}nalysis of {E}nd-to-{E}nd {I}nternet
  {D}ynamics,'' Ph.D. dissertation, University of California, Berkeley, 1997.

\bibitem{Moon:2000}
S.~B. Moon, ``Measurement and {A}nalysis of {E}nd-to-{E}nd {D}elay and {L}oss
  in the {I}nternet,'' Ph.D. dissertation, University of Massachusetts Amherst,
  2000.

\bibitem{T:2009}
J.~Tan, ``New {O}rigins of {H}eavy {T}ails with {A}pplications to {I}nformation
  {N}etworks,'' Ph.D. dissertation, Columbia University, 2009.

\end{thebibliography}

\end{document}